%% file: main.tex
\documentclass[11pt,a4paper]{article}

\usepackage[utf8]{inputenc}
\usepackage[T1]{fontenc}
\usepackage{lmodern}

\usepackage{amsmath,amssymb,amsthm,mathtools}
\usepackage{bm}

\usepackage[margin=2.6cm]{geometry}
\usepackage{microtype}
\usepackage{booktabs}
\usepackage{enumitem}
\usepackage{subcaption}
\usepackage{caption}

\usepackage{tikz}
\usepackage{pgfplots}
\pgfplotsset{compat=1.18}
\usepgfplotslibrary{groupplots}

\usepackage[colorlinks=true,allcolors=blue!55!black]{hyperref}
\usepackage[capitalize,noabbrev]{cleveref}

\usepackage{aliascnt}

\theoremstyle{plain}
\newtheorem{theorem}{Theorem}[section]
\newaliascnt{proposition}{theorem}
\newtheorem{proposition}[proposition]{Proposition}
\aliascntresetthe{proposition}
\newaliascnt{lemma}{theorem}
\newtheorem{lemma}[lemma]{Lemma}
\aliascntresetthe{lemma}
\newaliascnt{corollary}{theorem}

\aliascntresetthe{corollary}

\theoremstyle{definition}
\newaliascnt{definition}{theorem}
\newtheorem{definition}[definition]{Definition}
\aliascntresetthe{definition}
\newaliascnt{assumption}{theorem}
\newtheorem{assumption}[assumption]{Assumption}
\aliascntresetthe{assumption}
\newaliascnt{remark}{theorem}
\newtheorem{remark}[remark]{Remark}
\aliascntresetthe{remark}

\newcommand{\R}{\mathbb{R}}
\newcommand{\N}{\mathbb{N}}
\newcommand{\Q}{\mathcal{Q}}
\newcommand{\E}{\mathcal{E}}
\newcommand{\Diss}{\mathrm{Diss}}
\newcommand{\Dr}{\mathcal{D}}
\newcommand{\Var}{\mathrm{Var}}
\newcommand{\bs}[1]{\bm{#1}}
\newcommand{\dd}{\mathop{}\!\mathrm{d}}
\newcommand{\sgn}{\operatorname{sgn}}
\newcommand{\Id}{\mathrm{I}}
\DeclareMathOperator*{\argmin}{arg\,min}

\crefname{theorem}{Theorem}{Theorems}
\Crefname{theorem}{Theorem}{Theorems}
\crefname{proposition}{Proposition}{Propositions}
\Crefname{proposition}{Proposition}{Propositions}
\crefname{lemma}{Lemma}{Lemmas}
\Crefname{lemma}{Lemma}{Lemmas}
\crefname{corollary}{Corollary}{Corollaries}
\Crefname{corollary}{Corollary}{Corollaries}
\crefname{definition}{Definition}{Definitions}
\Crefname{definition}{Definition}{Definitions}
\crefname{remark}{Remark}{Remarks}
\Crefname{remark}{Remark}{Remarks}
\crefname{assumption}{Assumption}{Assumptions}
\Crefname{assumption}{Assumption}{Assumptions}

\begin{document}

\title{\bfseries Rate-Independent Epigenetics: a thermodynamically
consistent framework for modelling epigenetic response}

\author{Jacobo Ayensa-Jim\'enez\thanks{Technical University of Madrid
(UPM), Madrid, Spain. \texttt{jacobo.ayensa@upm.es}.}
\and
Ignacio Romero\thanks{Technical University of Madrid (UPM) and IMDEA
Materials Institute, Madrid, Spain. \texttt{ignacio.romero@upm.es}.}}

\date{\today}

\maketitle

\begin{abstract}
Epigenetic changes---heritable, long-lived, yet actively reversible
modifications of the chromatin state---display memory, threshold
activation and hysteresis, features that are the hallmark of
\emph{rate-independent} dissipative evolution. We propose a
mathematical framework, which we call Rate-Independent Epigenetics
(RIE), that models epigenetic change within the theory of
rate-independent systems and is, by construction, consistent with two fundamental principles that we identify with the laws of thermodynamics. In the proposed framework, a model is specified by a triple: a state
space of epigenetic configurations, a stored-energy functional
depending on the state and on an external (micro-environmental)
loading, and a $1$-homogeneous dissipation potential encoding the
resistance of the epigenetic machinery to change. Assuming an
energetic evolution principle, the governing equations---an energetic
(global stability plus energy balance) formulation and a resulting
subdifferential inclusion of Biot type---follow systematically, with
no further modelling hypotheses. Solutions satisfy, a priori, energetic
estimates that coincide with the proposed first and second laws: the energy
balance is exact energy conservation, and the $1$-homogeneity of the
dissipation potential forces a non-negative, minimal (\emph{economical})
dissipation. Under natural coercivity and continuity assumptions we
establish existence of energetic solutions and, via vanishing
viscosity, of balanced-viscosity solutions that resolve the ambiguity
of the energetic formulation at epigenetic switches; uniqueness holds
under convexity and, in the scalar case, under a mild
finite-multiplicity condition. We then build a variational time
integrator, prove its convergence to energetic solutions and a global
energy-consistency estimate, and reduce it, in the scalar case, to a
two-sided return map. The framework is illustrated on the scalar linear
play operator, the minimal RIE model of epigenetic memory, for which
the state evolution is obtained in closed form and reproduces both
reversible and irreversible (locked-in) responses; the integrator
recovers the exact solution at the nodes and exhibits the predicted
first-order energy consistency. We also validate the approach for an example with a double-well energetic potential, showing the ability of the framework to study multi-stability scenarios and catastrophic switches, ubiquitous in biological contexts and for a nonlinear problem, proving that the theoretical results hold. The presented framework can be seen as a skeleton for a richer thermodynamically-consistent theory incorporating viscous dissipation features.


\end{abstract}

\setcounter{tocdepth}{2}
\tableofcontents
\bigskip

\input{sections/introduction}
\input{sections/framework}
\input{sections/thermodynamics}
\input{sections/mathematical-results}
\input{sections/numerical-methods}
\input{sections/experiments}
\input{sections/discussion}
\input{sections/conclusions}

\appendix
\input{sections/appendix}

\bibliographystyle{plain}
\bibliography{references}

\end{document}

%% file: sections/introduction.tex
\section{Introduction}\label{sec:intro}

Epigenetic changes are heritable modifications of gene expression that
do not alter the DNA sequence, mediated for instance by DNA
methylation and histone modification~\cite{Bird2002}. Two features make
them distinctive from a modelling standpoint. First, they are
\emph{bistable}: cells commit to discrete phenotypes separated by
barriers, a picture formalised by Waddington's epigenetic
landscape~\cite{Waddington1957} and, in dynamical-systems terms, by
bistable regulatory switches~\cite{Huang2007,Ferrell2012}. Second, they
exhibit \emph{memory and hysteresis}: an epigenetic mark established in
response to a micro-environmental stimulus can persist after the
stimulus is removed, and its erasure requires crossing a distinct
threshold---as when hypoxia impairs the oxygen-dependent demethylation
machinery and leaves a durable methylation
imprint~\cite{Thienpont2016}. Threshold activation, path dependence and
residual state upon unloading are precisely the signatures of
\emph{rate-independent} dissipative evolution.

This observation is the starting point of the present work. We propose
to model epigenetic change within the mathematical theory of
rate-independent systems, a framework we call \emph{Rate-Independent
Epigenetics} (RIE). Rate-independent systems describe processes whose
response is invariant under time reparametrisation, governed by the
balance between a stored energy and a dissipation
mechanism~\cite{MielkeTheil,Mielke2005,MielkeRoubicek}. They were
developed largely in mechanics---plasticity~\cite{lubliner2008plasticity,
hanreddy2013}, viscoelasticity~\cite{tschoegl1989},
damage~\cite{lemaitre2005uh}, fracture~\cite{francfort1998,dalmaso2006},
diffusion in porous media~\cite{bear1972}, frictional
contact~\cite{wriggers2006}, phase
transformation~\cite{balljames1987,mielke2002,carstensen2002uh} and
ferromagnetism~\cite{desimone2002,brokate1996}---and possess a
variational structure that yields powerful existence, uniqueness and
approximation results, valid even in the non-smooth
setting~\cite{MielkeRoubicek,visintin1994,Krasnoselskii}. The central
premise of RIE is that this same structure is the natural language for
epigenetic evolution: the double-well landscape of Waddington becomes
the stored energy, the micro-environment becomes the loading, and the
resistance of the epigenetic machinery to remodelling becomes the
dissipation threshold. Even though hysteresis operators have been used
in mathematical biology ---for instance in population and
epidemiological models---never before with this energetic, variational and
thermodynamically consistent structure, and, to the best of our
knowledge, not for epigenetics: the novelty of RIE is thus twofold, in
the mathematical machinery it brings to bear and in the biological
domain to which that machinery is applied.

The appeal of this transfer is methodological as much as conceptual.
Once the three ingredients are fixed and an energetic evolution
principle is postulated, \emph{all} governing equations follow
systematically and are, by construction, thermodynamically consistent:
the first and second laws are not imposed a posteriori but are built
into the variational structure. This removes a recurrent difficulty in
quantitative epigenetics---the ad-hoc prescription of evolution
laws---and guarantees that any RIE model respects energy conservation
and non-negative entropy production irrespective of the choice of
landscape or dissipation.

Rate independence is, of course, an idealisation. We adopt it as the
\emph{quasi-static limit} of epigenetic evolution---the analogue of
perfect plasticity in mechanics---questionable in general, but a good
leading-order description whenever the chromatin state relaxes quickly
relative to its environment or stays frozen except at thresholds, and,
crucially, one that captures the threshold--memory--hysteresis
phenomenology while yielding a robust, thermodynamically consistent and
predictive skeleton. Rate-dependent effects, such as relaxation and
ratcheting, are recovered as a controlled viscous extension whose
vanishing-viscosity limit is the present model; we return to the
validity of the hypothesis, and to its relation with existing
paradigms, in \cref{sec:discussion}.

\paragraph{Scope and contributions.}
This paper develops the general RIE framework and its numerical
approximation. The starting point is precisely the definition of an epigenetic model: a triple $(\Q,E,\Psi)$ consisting of a state space
$\Q\subseteq\R^n$ of epigenetic configurations, a stored-energy
functional $E$ depending on the state and on a loading path $S(t)$, and
a convex, $1$-homogeneous dissipation potential $\Psi$. Based on this succinct postulate, our contributions are:
\begin{enumerate}[leftmargin=2em,itemsep=2pt]
  \item We derive the governing equations from the energetic principle:
  the energetic formulation (\cref{def:energetic}) and the equivalent
  Biot inclusion (\cref{prop:biot}), and give their thermodynamic
  reading---first law, second law, and a principle of minimal
  dissipation (\cref{sec:thermo}).
  \item We establish the well-posedness of RIE models
  (\cref{sec:math}): existence of energetic solutions, uniqueness under
  convexity, and existence and scalar uniqueness of balanced-viscosity
  solutions, which select the physically correct path at epigenetic
  switches.
  \item We construct and analyse a variational time integrator
  (\cref{sec:numerical}): convergence to energetic solutions and a
  global energy-consistency estimate, with the scalar update reducing
  to a two-sided return map.
  \item We validate the method on the linear play operator---the
  minimal RIE model of epigenetic memory---for which we give the exact
  closed-form evolution and reproduce both reversible and irreversible
  responses, as well as in a problem with a two-well energetic potential, covering bi-stability scenarios, the cornerstone of epigenetics. The method is also tested in a problem requiring nonlinear numerical solvers using the method of manufactured solutions to test the scope of the convergence and consistency results (\cref{sec:experiments}).
\end{enumerate}
To keep the narrative focused, the body states only the principal
results; supporting results and all proofs, in full detail, are
collected in the appendices.

\paragraph{Outline.}
\Cref{sec:framework} sets up the framework and derives the governing
equations. \Cref{sec:thermo} develops the thermodynamic interpretation.
\Cref{sec:math} presents the well-posedness theory and
\cref{sec:numerical} the numerical method and its analysis.
\Cref{sec:experiments} reports the numerical experiments;
\cref{sec:discussion} discusses the validity of the rate-independent
hypothesis, situates the framework among existing paradigms and
outlines its extensions; and \cref{sec:conclusions} concludes.

%% file: sections/framework.tex
\section{The rate-independent framework}\label{sec:framework}

\subsection{The defining triple}\label{sec:ingredients}

Following~\cite{MielkeRoubicek}, a rate-independent system is specified
by a triple $(\Q, E, \Psi)$.

\begin{enumerate}[label=(\roman*),leftmargin=2.2em,itemsep=3pt]
  \item \textbf{State space.} The set $\Q$ of \emph{admissible states} is a
  closed and convex subset of a Banach space; throughout we take
  \begin{equation}\label{eq:state_space}
    \Q \subseteq \R^n .
  \end{equation}

  \item \textbf{Loading.} The interaction of the system with its
  environment is encoded by a prescribed, time-dependent
  \emph{loading path}
  \begin{equation}
    S \in W^{1,1}(0,T;\R^{+}),
  \end{equation}
  with $W^{1,1}$ denoting the Sobolev space of integrable functions
  with integrable first derivative.
  
  \item \textbf{Stored energy.} The \emph{stored-energy functional}
  $E\colon \Q\times\R^{+}\to\R$ is assumed to admit the additive
  decomposition
  \begin{equation}\label{eq:energy}
    E(\bs q, S) \;=\; F(\bs q)\;-\;\bs q\cdot\bs\ell(S),
  \end{equation}
  where $F\colon\Q\to\R$ is the \emph{configurational free energy} and
  $\bs\ell\colon\R^{+}\to\R^n$ the \emph{interaction (loading)
  potential}. Of $F$ we require only the minimal conditions that
  thermodynamics places on a free energy: that it be \emph{lower
  semicontinuous} and \emph{bounded below}, so that $E(\cdot,S)$ is a
  well-defined energy admitting a ground state, and continuously
  differentiable wherever the driving force $\bs f=-\partial_{\bs q}E$
  is invoked. Convexity is deliberately \emph{not} assumed---it would
  exclude the bistable landscapes of interest---and the stronger
  coercivity and smoothness needed for the well-posedness results are
  collected separately in \cref{ass:standing}.

  \item \textbf{Dissipation potential.} The \emph{dissipation
  potential} is a convex, lower semicontinuous map
  $\Psi\colon\R^n\to\R$ satisfying
  \begin{itemize}[leftmargin=1.4em,itemsep=1pt]
    \item \emph{normalisation}: $\Psi(\bs 0)=0$;
    \item \emph{positiveness}: $\Psi(\bs v)\ge 0$ for all $\bs v\in\R^n$;
    \item \emph{$1$-homogeneity}:
    $\Psi(\lambda\bs v)=\lambda\,\Psi(\bs v)$ for all $\lambda>0$.
  \end{itemize}
  These are the structural assumptions required
  in~\cite{Mielke2005,MielkeRoubicek} for the dissipation potential of
  a rate-independent system; the $1$-homogeneity is the property that
  distinguishes rate-independent from viscous (rate-dependent)
  evolution. A map satisfying only the first two conditions is called
  a dissipation \emph{pseudo-potential}.
\end{enumerate}

The forthcoming well-posedness results require, in addition, the
quantitative hypotheses of \cref{ass:standing}, which are collected in
\cref{sec:math}. The benchmark of \cref{sec:experiments} satisfies all
of them.

\subsection{Derived canonical quantities}\label{sec:derived}

From the triple $(\Q,E,\Psi)$ the following objects are canonically
constructed.

\paragraph{(a) Thermodynamic driving force.}
When $E$ is differentiable in $\bs q$, the negative energy gradient
defines the \emph{generalised thermodynamic force}
\begin{equation}\label{eq:force}
  \bs f(\bs q, S) \;:=\; -\frac{\partial E}{\partial \bs q}(\bs q, S)
  \;=\; \bs\ell(S) - \bs g(\bs q),
  \qquad
  \bs g(\bs q) := \frac{\partial F}{\partial\bs q}(\bs q).
\end{equation}
The force $\bs f$ is the variable conjugate to $\bs q$: it quantifies
the thermodynamic drive toward evolution of the state.

\paragraph{(b) Elastic domain.}
The subdifferential of $\Psi$ at the origin,
\begin{equation}\label{eq:elastic_domain}
  \E := \partial\Psi(\bs 0),
\end{equation}
is the closed convex set of forces the system can sustain without
evolving. It is the \emph{elastic domain} in force space; by convex
duality~\cite{rockafellar1970,ekeland1999convex} $\Psi$ is the support
function of the elastic domain, $\Psi(\bs v)=\sup_{\bs f\in\E}\bs f\cdot\bs v$.

\paragraph{(c) Dissipation distance.}
The \emph{dissipation distance} is the minimal cost of a transition
between two states,
\begin{equation}\label{eq:diss_dist}
  \Dr(\bs q_1,\bs q_2)
  := \inf\Bigl\{\textstyle\int_0^1 \Psi(\dot\gamma(s))\dd s :
  \gamma\in\mathrm{AC}([0,1];\Q),\ \gamma(0)=\bs q_1,\ \gamma(1)=\bs q_2\Bigr\}.
\end{equation}

\paragraph{(d) Total dissipation.}
For a trajectory $\bs q\colon[0,T]\to\Q$ of bounded variation, the
\emph{total dissipation} is
\begin{equation}\label{eq:total_diss}
  \Diss_\Psi(\bs q;[0,t])
  := \sup\Bigl\{\textstyle\sum_{j=1}^{N}\Dr\bigl(\bs q(t_{j-1}),\bs q(t_j)\bigr):
  0=t_0<\dots<t_N=t,\ N\in\N\Bigr\}.
\end{equation}

\paragraph{(e) Power and work of the loading.}
The rate of energy change due to the loading at frozen state is
\begin{equation}\label{eq:power}
  \mathcal P_S(t) := \frac{\dd}{\dd t} E(\bs q, S(t))
  = \partial_S E(\bs q, S)\,\dot S(t)
  = -\bs q\cdot\bs\ell'(S(t))\,\dot S(t),
\end{equation}
and the associated \emph{work} performed on the system over $[0,t]$ is the path integral $\mathcal W_S(t):=\int_0^t\partial_S E(\bs q(s),S(s))\,\dot S(s)\dd s$.

\subsection{Governing equations: energetic formulation}
\label{sec:energetic}

We postulate the evolution through the following variational notion,
due to Mielke and Theil~\cite{MielkeTheil}.

\begin{definition}[Energetic solution]\label{def:energetic}
A function $\bs q\colon[0,T]\to\Q$ is an \emph{energetic solution} of the rate-independent system $(\Q,E,\Psi)$ for the loading $S$ if it satisfies:

\smallskip
\noindent\textbf{(S) Global stability.} For every $t\in[0,T]$ and every
$\tilde{\bs q}\in\Q$,
\begin{equation}\label{eq:stability}
  E(\bs q(t),S(t)) \;\le\; E(\tilde{\bs q},S(t)) + \Dr(\bs q(t),\tilde{\bs q}).
\end{equation}

\noindent\textbf{(E) Energy balance.} For every $t\in[0,T]$,
\begin{equation}\label{eq:energy_balance}
  E(\bs q(t),S(t)) + \Diss_\Psi(\bs q;[0,t])
  = E(\bs q(0),S(0)) + \int_0^t \partial_S E(\bs q(s),S(s))\,\dot S(s)\dd s.
\end{equation}
\end{definition}

Condition~(S) admits a variational reading: $\bs q(t)$ is a global
minimiser of the perturbed energy $\bs r\mapsto E(\bs r,S(t))+\Dr(\bs
q(t),\bs r)$. Condition~(E) states that, along the trajectory, the sum
of stored energy and cumulative dissipation equals the initial energy
plus the work supplied by the loading; it is an exact energy
conservation statement, examined in
\cref{sec:thermo}.

\subsection{Governing equations: subdifferential formulation}
\label{sec:subdiff}

Since $\Psi$ is convex but generally non-smooth at $\dot{\bs q}=\bs 0$,
the pointwise force balance takes the form of a differential inclusion.

\begin{proposition}[Biot inclusion]\label{prop:biot}
Let $\bs q\in\mathrm{AC}([0,T];\Q)$ satisfy \textnormal{(S)}
and~\textnormal{(E)}, with $E(\cdot,S)$ differentiable. Then, for a.e.\
$t\in[0,T]$,
\begin{equation}\label{eq:biot}
  \bs 0 \;\in\; \partial\Psi(\dot{\bs q}(t)) + \frac{\partial E}{\partial\bs q}(\bs q(t),S(t)),
\end{equation}
equivalently $\bs f(\bs q,S)\in\partial\Psi(\dot{\bs q})$. Conversely,
if $E(\cdot,S)$ is convex, then~\eqref{eq:biot} together with the
local stability $\bs f(\bs q(t),S(t))\in\E$ implies
\textnormal{(S)}--\textnormal{(E)}.
\end{proposition}

Equation~\eqref{eq:biot} is the (generalised) Biot
equation~\cite{Biot1956}; for $\Psi=\rho\|\cdot\|$ it coincides with the
sweeping process of Moreau~\cite{moreau1977} and, in one dimension, with
the play operator of hysteresis~\cite{Krasnoselskii,visintin1994}. Its
analysis rests on the \emph{saturation identity}, the algebraic
signature of $1$-homogeneity: if $\bs f\in\partial\Psi(\bs v)$ then
$\bs f\cdot\bs v=\Psi(\bs v)$ (\cref{lem:saturation} in
\cref{app:framework}). Combined with \cref{prop:biot}, it yields, along
any solution,
\begin{equation}\label{eq:saturation_solution}
  \bs f(\bs q,S)\cdot\dot{\bs q}=\Psi(\dot{\bs q}),
\end{equation}
i.e.\ the dissipation is exactly balanced by the energy release at
frozen loading. This is the cornerstone of the thermodynamic
interpretation of \cref{sec:thermo}. The proof of \cref{prop:biot} is
given in \cref{app:framework}.

%% file: sections/thermodynamics.tex
\section{Thermodynamic interpretation}\label{sec:thermo}

We now read the framework of \cref{sec:framework} in the language of
the thermodynamics of internal variables, identifying the energy
balance with the first law and the dissipation inequality with the
second law~\cite{germain1983,halphen1975generalized}.

\subsection{Free energy and the first law}\label{sec:first_law}

The state of the system is the pair $(\bs q, S)$, and $E$ plays the role
of a \emph{Helmholtz-type free energy}: the portion of the energy
stored in the configuration and recoverable through quasi-static
changes,
\begin{equation}\label{eq:free_energy}
  E(\bs q, S) = \underbrace{F(\bs q)}_{\text{configurational}}
  \;-\;\underbrace{\bs q\cdot\bs\ell(S)}_{\text{coupling to environment}}.
\end{equation}
The configurational term $F$ is intrinsic and independent of the
environment; the coupling term $-\bs q\cdot\bs\ell(S)$ acts as a
generalised chemical potential that biases the state. The conjugate
driving force is $\bs f=-\partial_{\bs q}E$, as in~\eqref{eq:force}, and
the power injected by the loading at frozen state is
$\mathcal P_S=\partial_S E\,\dot S$, as in~\eqref{eq:power}.

Rearranging the energy balance~\eqref{eq:energy_balance} over $[0,T]$
yields the \emph{first law}
\begin{equation}\label{eq:first_law}
  \underbrace{\Delta E}_{\text{change in free energy}}
  \;+\;
  \underbrace{\Diss_\Psi(\bs q;[0,T])}_{\mathcal D_{\mathrm{tot}}\ \ge 0}
  \;=\;
  \underbrace{\int_0^T \partial_S E(\bs q,S)\,\dot S\dd t}_{\mathcal W\ \text{(work of the environment)}},
\end{equation}
with $\Delta E=E(\bs q(T),S(T))-E(\bs q(0),S(0))$: the work supplied by
the environment is stored as free energy or irreversibly dissipated.

\subsection{Entropy production and the second law}\label{sec:second_law}

The instantaneous \emph{dissipation rate} is
\begin{equation}\label{eq:diss_rate}
  \mathcal D(t) := \bs f(\bs q(t),S(t))\cdot\dot{\bs q}(t).
\end{equation}

\begin{proposition}[Second law]\label{prop:second_law}
Along any solution of~\eqref{eq:biot} one has $\mathcal D(t)\ge 0$ for
a.e.\ $t$. Moreover the \emph{excess dissipation}
$\Xi(t):=\bs f(t)\cdot\dot{\bs q}(t)-\Psi(\dot{\bs q}(t))$ vanishes
identically, $\Xi\equiv 0$.
\end{proposition}

This is an immediate consequence of the saturation
identity~\eqref{eq:saturation_solution} (see \cref{app:framework}).
The non-negativity $\mathcal D\ge 0$ is the second law: dissipation is
non-negative regardless of the direction of evolution. The identity
$\Xi\equiv 0$ expresses a \emph{principle of minimal (economical)
dissipation}: a rate-independent system saturates the dissipation
inequality, dissipating exactly the minimum compatible with the energy
balance, with no viscous excess. In an isothermal setting at absolute
temperature $\Theta$, the irreversible entropy production rate is
$\Theta\,\dot s_{\mathrm{irr}}=\mathcal D(t)\ge0$, so that over the
whole process
\begin{equation}\label{eq:entropy}
  \Delta s_{\mathrm{irr}} = \frac{\mathcal D_{\mathrm{tot}}}{\Theta}
  = \frac{\Diss_\Psi(\bs q;[0,T])}{\Theta}\ \ge\ 0,
\end{equation}
the standard form of the second law.

\subsection{Geometric dissipation over closed cycles}\label{sec:hysteresis}

A hallmark of rate-independent thermodynamics is that the energy
dissipated over a closed loading cycle depends only on the geometry of
the trajectory, not on its time scale. Consider a cycle with
$S(0)=S(T)=S_0$ that returns the state to its initial value, so that
$\Delta E=0$. Integrating~\eqref{eq:first_law} and using
$\mathcal D\ge0$ (\cref{prop:second_law}) gives
\begin{equation}\label{eq:cycle}
  \mathcal W = \oint_0^T \partial_S E(\bs q,S)\,\dot S\dd t
  = \mathcal D_{\mathrm{tot}} \ \ge\ 0.
\end{equation}
Hence the work performed by the environment over a closed cycle cannot
be recovered and is entirely converted into irreversible entropy
production. In the $(\bs q,\bs f)$ plane this manifests as a
non-vanishing \emph{hysteresis loop} whose enclosed area equals exactly
$\mathcal D_{\mathrm{tot}}$; since $\bs f=\bs\ell(S)-\bs g(\bs q)$
(\cref{eq:force}), the same loop, up to the fixed reparametrisation
$\bs f\leftrightarrow\bs\ell(S)$, is the one plotted directly in the
$(\ell,q)$ loading--state plane in the numerical experiments of
\cref{sec:experiments}. The rate-independent structure therefore
guarantees that returning the environment to its baseline leaves a
permanent thermodynamic imprint---a property that, in the motivating
biological application, models the non-volatility of the epigenetic
state.

%% file: sections/mathematical-results.tex
\section{Well-posedness}\label{sec:math}

We now state the principal existence and uniqueness results for RIE
models. Throughout we work under the following standing hypotheses; the
supporting results (properties of the dissipation distance and energy
estimates in \cref{app:framework}, compatibility conditions in
\cref{app:math}) and all proofs are collected, in detail and in logical
order, in the appendices.

\begin{assumption}[Standing hypotheses]\label{ass:standing}
The triple $(\Q,E,\Psi)$ with $E(\bs q,S)=F(\bs q)-\bs q\cdot\bs\ell(S)$
satisfies:
\begin{enumerate}[label=\textnormal{(A\arabic*)},leftmargin=2.6em,itemsep=2pt]
  \item\label{A:W} \emph{Stored energy.} $\Q\subseteq\R^n$ is closed and
  $F\colon\Q\to\R$ is continuous and coercive of order $p>1$: there
  exist $c_1>0$, $c_2\ge0$ with $F(\bs q)\ge c_1\|\bs q\|^p-c_2$.
  \item\label{A:ell} \emph{Interaction potential.}
  $\bs\ell\colon[0,\infty)\to\R^n$ is Lipschitz and bounded, with
  $C_\ell:=\sup_{S\ge0}\|\bs\ell(S)\|<\infty$ and
  $C_{\ell'}:=\sup_{S\ge0}\|\bs\ell'(S)\|<\infty$.
  \item\label{A:S} \emph{Loading.} $S\in W^{1,1}(0,T;\R^+)$.
  \item\label{A:Psi} \emph{Dissipation potential.}
  $\Psi\colon\R^n\to[0,\infty)$ is continuous, convex, normalised
  $(\Psi(\bs 0)=0)$, positively $1$-homogeneous and symmetric
  $(\Psi(-\bs v)=\Psi(\bs v))$.
  \item\label{A:coerc} \emph{Definiteness of dissipation.} There
  exists $c_\Psi>0$ with $\Psi(\bs v)\ge c_\Psi\|\bs v\|$ for all
  $\bs v$; equivalently, $\bs 0\in\operatorname{int}\E$.
\end{enumerate}
\end{assumption}

Hypotheses \ref{A:W}--\ref{A:Psi} are the classical ingredients for
energetic solutions; \ref{A:coerc}, which we add explicitly, is needed
for the total-variation bound underlying balanced-viscosity solutions
and holds with $c_\Psi=\rho$ for $\Psi=\rho\|\cdot\|$. In the
epigenetic reading, \ref{A:W} is the confinement of the chromatin state
by a coercive landscape, \ref{A:ell} is the boundedness of the
micro-environmental drive, \ref{A:S} is the regularity (integrability) of the microenvironment signal, \ref{A:Psi} is the dissipative (although rate-independent) character of biomolecular chromatin machinery (here assumed as symmetric), and \ref{A:coerc} is the strict positivity of
the resistance to remodelling.

\subsection{Existence and uniqueness of energetic solutions}
\label{sec:existence}

The stability test~\eqref{eq:stability} singles out the following set.

\begin{definition}[Stability set]\label{def:stability_set}
For $t\in[0,T]$, $\mathcal S(t):=\{\bs q\in\Q: E(\bs q,S(t))\le
E(\tilde{\bs q},S(t))+\Dr(\bs q,\tilde{\bs q})\ \ \forall\,\tilde{\bs q}\in\Q\}$.
\end{definition}

\begin{theorem}[Existence]\label{thm:existence}
Under \ref{A:W}--\ref{A:Psi}, for every initial datum
$\bs q_0\in\mathcal S(0)$ there exists at least one energetic solution
$\bs q\colon[0,T]\to\Q$ of bounded variation, in the sense of
\cref{def:energetic}.
\end{theorem}

Energetic solutions need not be unique: at an epigenetic switch,
several post-jump states may satisfy~(S)--(E). Uniqueness is recovered
under convexity of the energetic landscape.

\begin{theorem}[Uniqueness under convexity]\label{thm:uniqueness}
Assume, in addition to \ref{A:W}--\ref{A:Psi}, that $F\in C^3(\Q)$ is
uniformly convex ($D^2F\succeq\alpha\Id$, $\alpha>0$) and that
$\bs\ell\in C^3$, $S\in C^3([0,T])$. Then for every
$\bs q_0\in\mathcal S(0)$ the energetic solution is unique.
\end{theorem}

Uniform convexity forces a single well and thus excludes the bistable
landscapes central to epigenetics; the balanced-viscosity theory below
restores uniqueness without convexity.

\subsection{Balanced-viscosity solutions}\label{sec:bv}

At a switch $t^\ast$, conditions~(S)--(E) require the released energy to
cover the dissipation cost $\Dr(\bs q^-(t^\ast),\bs q^+(t^\ast))$ but do
not select the post-switch state. Balanced-viscosity (BV)
solutions~\cite{mielke2012bv} resolve this by selecting the transition
path as the vanishing-viscosity limit ($\varepsilon\to0^+$) of the
regularised inclusion
\begin{equation}\label{eq:regularised}
  \bs 0\in\partial\Psi(\dot{\bs q}_\varepsilon)
  +\varepsilon\,\dot{\bs q}_\varepsilon
  +\frac{\partial E}{\partial\bs q}(\bs q_\varepsilon,S),
\end{equation}
whose memory is retained in a correction concentrated at the switches.

\begin{definition}[Balanced-viscosity solution]\label{def:BV}
A function $\bs q\colon[0,T]\to\Q$ of bounded variation is a
\emph{BV solution} of $(\Q,E,\Psi)$ if:

\smallskip
\noindent\textbf{(BV-S) Local stability.} For a.e.\ $t$, there is
$\delta>0$ such that $E(\bs q(t),S(t))\le E(\bs r,S(t))+\Dr(\bs q(t),\bs r)$
for all $\bs r\in B_\delta(\bs q(t))\cap\Q$.

\smallskip
\noindent\textbf{(BV-E) Energy balance with viscous correction.}
For every $t$,
\begin{equation}\label{eq:BV_energy}
  E(\bs q(t),S(t))+\Diss_\Psi(\bs q;[0,t])+\mathcal V(\bs q;[0,t])
  = E(\bs q_0,S(0))+\int_0^t\partial_S E(\bs q,S)\dot S\dd s,
\end{equation}
with $\mathcal V(\bs q;[0,t]):=\sum_{s\le t}
[E(\bs q^-(s),S(s))-E(\bs q^+(s),S(s))-\Dr(\bs q^-(s),\bs q^+(s))]^+$,
summed over the (at most countably many) jump points.
\end{definition}

\begin{theorem}[Existence of BV solutions]\label{thm:BV_existence}
Under \ref{A:W}--\ref{A:coerc}, for every $\bs q_0\in\mathcal S(0)$
there exists at least one BV solution in the sense of \cref{def:BV}.
\end{theorem}

\begin{theorem}[Scalar uniqueness of BV solutions]\label{thm:BV_uniqueness}
Assume, in addition to \ref{A:W}--\ref{A:coerc}, that $n=1$ and that
$q\mapsto E(q,S)$ has at most finitely many local minima for each
$S\ge0$. Then the BV solution is unique.
\end{theorem}

This is the case of genuine epigenetic interest. The bistable
landscapes that model competing phenotypes are, by construction,
non-convex, so the convexity hypothesis of \cref{thm:uniqueness} is
never met and energetic uniqueness does \emph{not} apply; in these
models uniqueness rests \emph{exclusively} on the balanced-viscosity
selection of \cref{thm:BV_uniqueness}. Since every energetic solution is
a BV solution, uniqueness of BV solutions moreover implies that the
energetic solution, when it is of bounded variation, is the unique one.

%% file: sections/numerical-methods.tex
\section{Variational numerical method}\label{sec:numerical}

\subsection{Incremental minimisation and the return map}
\label{sec:incremental}

We approximate energetic solutions by incremental
minimisation~\cite{ortiz1999tq,mielke2008zb,conti2008,romero2021dd}.
Partition $[0,T]$ as $0=t_0<\dots<t_N=T$, write $h_k=t_{k+1}-t_k$,
$S_k=S(t_k)$, and, given $\bs q_k$, define
\begin{equation}\label{eq:Ikstar}
  \bs q_{k+1}\in\argmin_{\bs q\in\Q}
  \bigl[E(\bs q,S_{k+1})+\Dr(\bs q_k,\bs q)\bigr].
\end{equation}
The optimality condition of~\eqref{eq:Ikstar} is the backward-Euler
discretisation of the Biot inclusion~\eqref{eq:biot},
\begin{equation}\label{eq:discrete_inclusion}
  \bs 0\in\frac{\partial E}{\partial\bs q}(\bs q_{k+1},S_{k+1})
  +\partial\Psi(\bs q_{k+1}-\bs q_k),
\end{equation}
which splits into an \emph{elastic lock}
($\bs q_{k+1}=\bs q_k$ when $-\partial_{\bs q}E(\bs q_k,S_{k+1})\in\E$)
and a \emph{dissipative flow} otherwise (\cref{prop:return_map} in
\cref{app:numerical}). This predictor--corrector structure---an elastic
trial state, corrected by projection onto the elastic domain $\E$
whenever it turns out inadmissible---is precisely the
\emph{return-mapping} algorithm of computational
plasticity~\cite{Simo1986,hanreddy2013}, and holds for \emph{any}
$\Psi$. What is special about the scalar case $\Psi(\dot q)=\rho|\dot q|$
is only that $\E=[-\rho,\rho]$ is an interval and $g$ is
scalar-invertible, so that the correction takes the explicit
\emph{two-sided} closed form: with the trial force
$f^{\mathrm{tr}}=\ell(S_{k+1})-g(q_k)$,
\begin{equation}\label{eq:scalar_return_map}
  q_{k+1}=
  \begin{cases}
    q_k, & |f^{\mathrm{tr}}|\le\rho, \\[2pt]
    g^{-1}\!\bigl(\ell(S_{k+1})-\sgn(f^{\mathrm{tr}})\,\rho\bigr),
      & |f^{\mathrm{tr}}|>\rho,
  \end{cases}
\end{equation}
an elastic predictor followed, upon violation of the yield condition
$|f|=\rho$, by a corrector onto the yield surface. The scheme is
unconditionally stable in the energetic sense and satisfies a discrete
energy inequality (\cref{prop:discrete_energy}).

\subsection{Convergence}\label{sec:convergence}

\begin{theorem}[Convergence of the incremental scheme]\label{thm:convergence}
Under \ref{A:W}--\ref{A:Psi}, let $\bs q_0\in\mathcal S(0)$ and let
$\bs q^h$ be the piecewise-constant interpolant of the sequence
generated by~\eqref{eq:Ikstar} on a partition of maximal step $h$.
Then, as $h\to0$, a subsequence $\bs q^{h_j}(t)\to\bs q(t)$ for every
$t\in[0,T]$, where $\bs q$ is an energetic solution of $(\Q,E,\Psi)$.
If the solution is unique (\cref{thm:uniqueness} or the scalar case of
\cref{thm:BV_uniqueness}), the whole family converges.
\end{theorem}

\subsection{Consistency}\label{sec:consistency}

Substituting the exact solution into the discrete energy inequality
produces, since the exact solution obeys the continuous
balance~\eqref{eq:energy_balance}, the \emph{local truncation error}
\begin{equation}\label{eq:truncation}
  \tau_k := \int_{t_k}^{t_{k+1}}\!\partial_S E(\bs q(s),S(s))\,\dot S(s)\dd s
  -\bigl[E(\bs q(t_k),S_{k+1})-E(\bs q(t_k),S_k)\bigr],
\end{equation}
i.e.\ the gap between the continuous work and its frozen-state discrete
approximation. It is locally of first order for absolutely continuous
solutions and of second order for Lipschitz (equivalently,
$W^{1,\infty}$) ones (\cref{prop:consistency}); accumulated over the
partition, it yields the main consistency estimate.

\begin{theorem}[Global energy consistency]\label{thm:global_consistency}
Let $h=\max_k h_k$ and $\mathcal E_\tau:=\sum_{k=0}^{N-1}|\tau_k|$.
Under \ref{A:ell}, $S\in W^{1,\infty}(0,T)$ and either
$\bs q\in\mathrm{AC}$ or $\bs q\in W^{1,\infty}$, one has
\begin{equation}\label{eq:global_cons}
  \mathcal E_\tau\le C\,h,
\end{equation}
with $C=C_{\ell'}\|\dot S\|_{L^\infty}\int_0^T\|\dot{\bs q}\|\dd\tau$ in
the first case and $C=\tfrac12 C_{\ell'}\|\dot{\bs q}\|_{L^\infty}
\|\dot S\|_{L^\infty}T$ in the second. Thus the scheme is globally
first-order consistent in the energy balance, irrespective of the
regularity of $\bs q$.
\end{theorem}

The regularity of $\bs q$ is inherited from that of the data $F$,
$\bs\ell$, $S$: away from switches $\bs q$ is Lipschitz and
\cref{thm:global_consistency} applies. At a switch $\bs q$ is only of
bounded variation and the theorem's regularity hypothesis fails; the
accumulated defect is nonetheless found to remain $O(h)$ in the bistable
experiment of \cref{sec:conv2}, for the structural reason (a
dissipation-neutral energetic jump) explained there. All proofs of this
section, and the supporting discrete stability and local-consistency
results, are given in \cref{app:numerical}.

\begin{remark}[Nodal exactness and the observed state convergence]
\label{rem:nodal}
\Cref{thm:convergence} gives convergence of the state without a rate,
and \cref{thm:global_consistency} concerns the energy balance, not the
state; yet the scalar benchmarks of \cref{sec:experiments} display a
clean first-order convergence of the state in norm. This is not a
further a priori estimate but a consequence of a property those examples
satisfy exactly: the integrator reproduces the solution \emph{at the
nodes}, $q^h(t_i)=q(t_i)$. Indeed, on any interval of monotone loading
the return map solves the same algebraic yield condition
$g(q_{k+1})=\ell(S_{k+1})\mp\rho$ that the exact solution obeys at
$t_{k+1}$, and reproduces the constant lock values exactly, so the state
carries no accumulated truncation error---which is precisely why the
genuine discretisation error lives in the energy balance
(\cref{thm:global_consistency}). Under this nodal exactness, and for a
piecewise-Lipschitz solution with finitely many jumps, the norm error of
the piecewise-constant interpolant is a pure zeroth-order-hold
interpolation error, hence $O(h)$: this is made precise and proved in
\cref{prop:interp}. We stress that nodal exactness is specific to these
scalar, locally monotone settings and is \emph{not} a general property
of the scheme, so the observed first-order rate must not be read as a
general a priori estimate.
\end{remark}

%% file: sections/experiments.tex
\section{Numerical experiments}\label{sec:experiments}

We validate the method on three benchmarks, all admitting a closed-form
solution so that the integrator can be compared directly against exact
values: the scalar \emph{linear play operator}
(\cref{sec:benchmark,sec:regimes,sec:conv_exp}), the minimal RIE model
of epigenetic memory, whose landscape is convex; a
\emph{piecewise-quadratic double well} (\cref{sec:benchmark2}), which
is genuinely bistable and excites the balanced-viscosity mechanism of
\cref{sec:bv}; and a \emph{nonlinear convex} problem
(\cref{sec:benchmark3}) whose transcendental force makes the return map
require a genuine iterative solve at every dissipative step. All curves
shown are the \emph{numerical} solution produced by the return-map
integrator; we report separately, in each convergence study, how closely
it reproduces the closed-form values.

\subsection{First minimal example}\label{sec:benchmark1}
\subsubsection{The linear play operator}\label{sec:benchmark}

Let $\Q=\R$ and
\begin{equation}\label{eq:benchmark}
  F(q)=\tfrac12 k\,q^2,\qquad
  \ell(S)=\ell_\infty\bigl(1-e^{-\lambda S}\bigr),\qquad
  \Psi(\dot q)=\rho\,|\dot q|,
\end{equation}
with $k,\ell_\infty,\lambda,\rho>0$, so that $g(q)=F'(q)=k\,q$ is linear
and the driving force is $f(q,S)=\ell(S)-k\,q$. The landscape is convex
(a single well), the loading is increasing, bounded
($C_\ell=\ell_\infty$) and Lipschitz ($C_{\ell'}=\ell_\infty\lambda$),
and $\Psi=\rho|\cdot|$ is definite with $c_\Psi=\rho$; hence
\ref{A:W}--\ref{A:coerc} all hold and the solution is unique
(\cref{thm:uniqueness}).

Because $g$ is linear, the yield condition is inverted explicitly:
setting $u(t)=\ell(S(t))/k$ and $r=\rho/k$, the state is the classical
play operator $q=\mathrm{P}_r[u]$ with $q(0)=0$. For a loading cycle
$S(t)=\tfrac12 S_{\max}(1-\cos(2\pi t/T))$ that increases from $0$ to
$S_{\max}$ and back---so that $u$ rises on $[0,T/2]$ and falls on
$[T/2,T]$---the evolution is, in closed form,
\begin{equation}\label{eq:play_exact}
  q(t)=
  \begin{cases}
    \max\{0,\;u(t)-r\}, & 0\le t\le T/2 \quad(\text{loading}),\\[4pt]
    \min\{q_{\max},\;u(t)+r\}, & T/2< t\le T \quad(\text{unloading}),
  \end{cases}
\end{equation}
with peak damage $q_{\max}=\max\{0,\,u_{\max}-r\}$ and
$u_{\max}=\ell(S_{\max})/k$. The residual state is
\begin{equation}\label{eq:residual}
  q_{\mathrm{res}}=q(T)=
  \begin{cases}
    \rho/k, & \ell_{\max}>2\rho \quad(\text{reversible}),\\[2pt]
    q_{\max}, & \tfrac12\ell_{\max}<\rho<\ell_{\max} \quad(\text{irreversible}),
  \end{cases}
\end{equation}
where $\ell_{\max}=\ell(S_{\max})$. The dissipation threshold $\rho$ ---
the resistance of the epigenetic machinery --- thus governs a sharp
dichotomy: a weak resistance ($\rho<\ell_{\max}/2$) allows partial
erasure of the mark upon removal of the stimulus, whereas a strong
resistance ($\rho>\ell_{\max}/2$) locks the mark in permanently.

\subsubsection{Reversible and irreversible responses}\label{sec:regimes}

We fix $k=1$, $\ell_\infty=0.5$, $\lambda=1$, $S_{\max}=5$
($\ell_{\max}\approx0.497$), $T=1$, and contrast the two regimes:
a \emph{reversible} one, $\rho=0.1<\ell_{\max}/2$, and an
\emph{irreversible} one, $\rho=0.3\in(\ell_{\max}/2,\ell_{\max})$.
\Cref{fig:time} shows the \emph{numerical} time evolution (return-map
integrator, $N=2000$): in both regimes the state stays locked until the
drive reaches the yield threshold, then follows the loading; on
unloading it locks again and either recovers to
$q_{\mathrm{res}}=\rho/k=0.1$ (reversible) or remains frozen at
$q_{\max}\approx0.197$ (irreversible). \Cref{fig:loops} shows the
corresponding numerical hysteresis loops, whose area equals the
dissipation per cycle. As established quantitatively in
\cref{sec:conv_exp}, these numerical curves are visually
indistinguishable from the closed-form
solution~\eqref{eq:play_exact}--\eqref{eq:residual} because, for this
benchmark, the integrator reproduces it \emph{exactly} at every node
(\cref{tab:val}).

\begin{figure}[t]
  \centering
  \begin{subfigure}[t]{0.47\textwidth}
    \centering
    \begin{tikzpicture}
      \begin{axis}[width=\textwidth,height=0.72\textwidth,
        xmin=0,xmax=1,axis y line*=left,xlabel={$t$},ylabel={$q(t)$},
        ymin=-0.02,ymax=0.7,tick label style={font=\small}]
        \addplot[thick,blue] table[x=t,y=q]{figures/fig_time_rev.dat};
      \end{axis}
      \begin{axis}[width=\textwidth,height=0.72\textwidth,
        xmin=0,xmax=1,axis y line*=right,axis x line=none,ylabel={$S(t)$},
        tick label style={font=\small},legend pos=north west,
        legend cell align=left]
        \addlegendimage{thick,blue}\addlegendentry{$q(t)$}
        \addplot[thick,red!70!black,dashed] table[x=t,y=S]{figures/fig_time_rev.dat};
        \addlegendentry{$S(t)$}
      \end{axis}
    \end{tikzpicture}
    \caption{Reversible ($\rho=0.1$)}
  \end{subfigure}
  \hfill
  \begin{subfigure}[t]{0.47\textwidth}
    \centering
    \begin{tikzpicture}
      \begin{axis}[width=\textwidth,height=0.72\textwidth,
        xmin=0,xmax=1,axis y line*=left,xlabel={$t$},ylabel={$q(t)$},
        ymin=-0.02,ymax=0.7,tick label style={font=\small}]
        \addplot[thick,blue] table[x=t,y=q]{figures/fig_time_irr.dat};
      \end{axis}
      \begin{axis}[width=\textwidth,height=0.72\textwidth,
        xmin=0,xmax=1,axis y line*=right,axis x line=none,ylabel={$S(t)$},
        tick label style={font=\small},legend pos=north west,
        legend cell align=left]
        \addlegendimage{thick,blue}\addlegendentry{$q(t)$}
        \addplot[thick,red!70!black,dashed] table[x=t,y=S]{figures/fig_time_irr.dat};
        \addlegendentry{$S(t)$}
      \end{axis}
    \end{tikzpicture}
    \caption{Irreversible ($\rho=0.3$)}
  \end{subfigure}
  \caption{Time evolution of the epigenetic state $q(t)$ under the
  cyclic loading $S(t)$. (a) Weak resistance: the mark is partially
  erased on unloading. (b) Strong resistance: the mark is locked in and
  the state remains frozen at its peak.}
  \label{fig:time}
\end{figure}
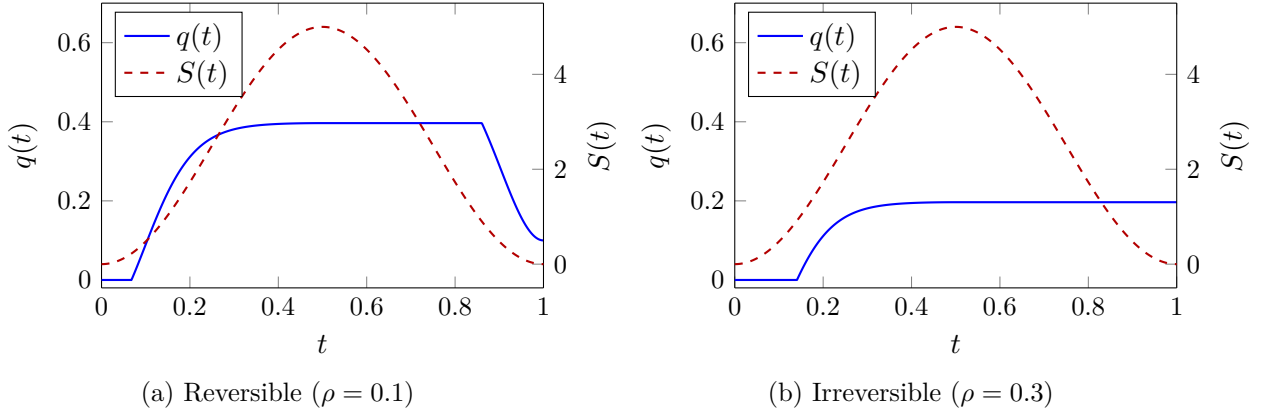

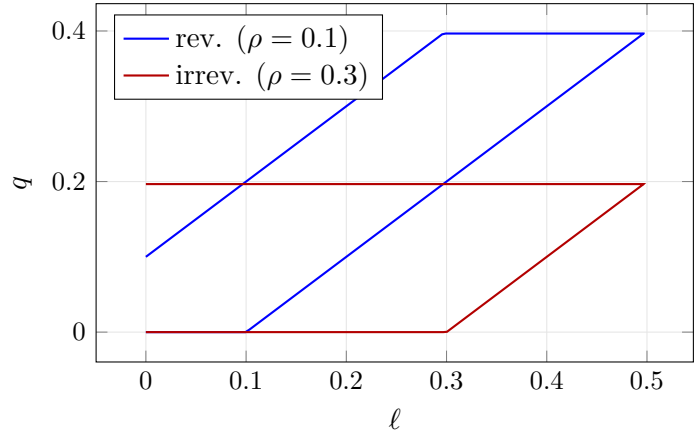
\begin{figure}[t]
  \centering
  \begin{tikzpicture}
    \begin{axis}[width=0.6\textwidth,height=0.4\textwidth,
      xlabel={$\ell$},ylabel={$q$},legend pos=north west,
      legend cell align=left,grid=both,grid style={gray!20},
      tick label style={font=\small}]
      \addplot[thick,blue] table[x=ell,y=q]{figures/fig_loop_rev.dat};
      \addplot[thick,red!70!black] table[x=ell,y=q]{figures/fig_loop_irr.dat};
      \legend{rev. ($\rho=0.1$),irrev. ($\rho=0.3$)}
    \end{axis}
  \end{tikzpicture}
  \caption{Hysteresis loops in the $(\ell,q)$ plane. The enclosed area
  is the dissipation per cycle. Larger resistance $\rho$ widens the
  elastic range and, beyond $\rho=\ell_{\max}/2$, suppresses recovery.}
  \label{fig:loops}
\end{figure}

\begin{table}[t]
  \centering
  \caption{Closed-form~\eqref{eq:play_exact}--\eqref{eq:residual} vs.\
  computed extrema ($N=4000$).}
  \label{tab:val}
  \begin{tabular}{lcccc}
    \toprule
    regime & $q_{\max}$ (exact / num.) & $q_{\mathrm{res}}$ (exact / num.) \\
    \midrule
    reversible ($\rho=0.1$)   & $0.39663$ / $0.39663$ & $0.10000$ / $0.10000$ \\
    irreversible ($\rho=0.3$) & $0.19663$ / $0.19663$ & $0.19663$ / $0.19663$ \\
    \bottomrule
  \end{tabular}
\end{table}

\subsubsection{Convergence and consistency}\label{sec:conv_exp}

We verify the consistency theory of \cref{sec:consistency} on the
reversible cycle, comparing against the exact
solution~\eqref{eq:play_exact}. Two facts emerge, both consistent with
the theory.

First, the \emph{state} error $\|q^h-q\|_\infty$ is at the level of
machine precision at every step size tested. This is a structural
property of rate-independent evolution: during flow the return
map~\eqref{eq:scalar_return_map} solves the algebraic yield condition
$k\,q_{k+1}=\ell(S_{k+1})\mp\rho$ \emph{exactly} at the node
$t_{k+1}$, independently of $h$, and during lock $q$ is exactly
constant; since the loading is monotone between nodes, the discrete
play operator coincides with the continuous one at the nodes. This
\emph{nodal exactness} is what underlies the first-order convergence of
the state in norm reported for the next two benchmarks; as explained in
\cref{rem:nodal} and proved in \cref{prop:interp}, it is a property
genuine of these scalar, monotonically loaded examples---not a general
convergence rate of the scheme. Here it makes the integrator reproduce
the exact evolution outright, so the discretisation error is carried
entirely by the energy balance rather than by the state.

Second, the global energy-balance defect $\sum_k e_k=(E_0+\mathcal
W_h)-(E_N+\Diss_h)\ge0$---the accumulated gap in the discrete energy
inequality---decays as $O(h)$ (\cref{fig:convergence}). This slope-one
behaviour is the global energy consistency of
\cref{thm:global_consistency}, whose hypotheses are met here ($q$ is
Lipschitz and $S\in W^{1,\infty}$): halving $h$ halves the defect, and
the data track the reference slope-one line.

\begin{figure}[t]
  \centering
  \begin{tikzpicture}
    \begin{loglogaxis}[width=0.6\textwidth,height=0.4\textwidth,
      xlabel={$h$},ylabel={$\sum_k e_k$ (energy defect)},
      legend pos=north west,legend cell align=left,
      grid=both,grid style={gray!20},tick label style={font=\small}]
      \addplot[only marks,mark=square*,teal!70!black] table[x=h,y=energy_defect]{figures/fig_convergence.dat};
      \addplot[dashed,black,domain=1.5e-4:2e-2]{0.754*x};
      \legend{energy defect,slope $1$}
    \end{loglogaxis}
  \end{tikzpicture}
  \caption{Global energy-balance defect versus step size $h$
  (log--log), confirming the first-order estimate of
  \cref{thm:global_consistency}. The state error, by contrast, is at
  machine precision because the return map is exact at the nodes.}
  \label{fig:convergence}
\end{figure}
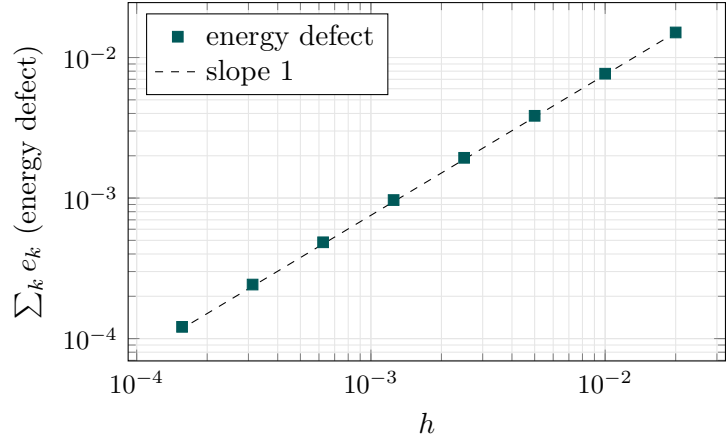

\subsubsection{A loading with sub-step structure}\label{sec:conv_osc}

The clean first-order convergence just reported hinges on nodal exactness
(\cref{rem:nodal}), which requires that the exact state not change regime
more than once within a step. We close this benchmark by deliberately
breaking that condition: keeping the same convex model, we drive it with
a loading that carries genuine structure on a short time scale $\tau$,
\begin{equation}\label{eq:osc_loading}
  \ell(t)=\ell_0+\alpha\,t+\beta\sin(2\pi t/\tau)
  \qquad(\text{by taking }S(t)=t),
\end{equation}
with $\ell_0=0$, $\alpha=0.3$, $\beta=0.15>\rho=0.1$ and $\tau=0.05$.
Because $\beta>\rho$, the oscillation repeatedly crosses the yield
surface, so the exact play-operator response \emph{chatters}---entering
and leaving flow---on the scale $\tau$ (\cref{fig:time_osc}). A step
$h\gg\tau$ cannot see this: the return map, reading only the endpoints of
each step, misses the sub-step excursions, nodal exactness fails, and the
discrete trajectory follows a genuinely different path.

\Cref{fig:conv_osc_state} shows the consequence for the state error. It
has a \emph{plateau}: for $h\gtrsim\tau$ it is $O(1)$, essentially independent
of $h$ (the scheme selects the wrong branch), and only once $h$ falls
below $\approx\tau/4$ does nodal exactness set in and the error collapse
onto the $O(h)$ interpolation regime of \cref{prop:interp}. Above that
threshold no a priori bound protects the state; its first-order rate is
genuine of the resolved regime, exactly as \cref{prop:interp} and
\cref{rem:interp_scope} anticipate.

\Cref{fig:conv_osc_energy} shows the sharply different fate of the energy
defect. \Cref{thm:global_consistency} bounds it by $C\,h$ with
$C=\tfrac12 C_{\ell'}\|\dot q\|_\infty\|\dot S\|_\infty T$;
since $q$ stays Lipschitz---with a large constant
$\|\dot q\|_\infty\sim\beta/\tau$---this bound holds at \emph{every} step
size, and numerically the defect never exceeds $0.14\,C\,h$ across the
whole sweep. Strikingly, at the coarse steps where the state error is
\emph{largest} the energy defect is \emph{smallest} (the under-resolved
scheme dissipates little): the two do not fail together. As $h$ crosses
$\tau/4$ the defect settles onto its $O(h)$ asymptote, the stiffness
$1/\tau$ having been absorbed into the constant $C$.

This is the sharpest illustration of the distinction running through the
paper. The energy consistency is a \emph{proven, unconditional}
statement---the bound $C\,h$ never fails, the resolution scale entering
only through its constant---whereas the first-order convergence of the
\emph{state} is not proven in general and, as here, genuinely can fail
until the discretisation resolves the physics. This is precisely why the
energy estimate is the robust result and the state rate is treated as
conditional.

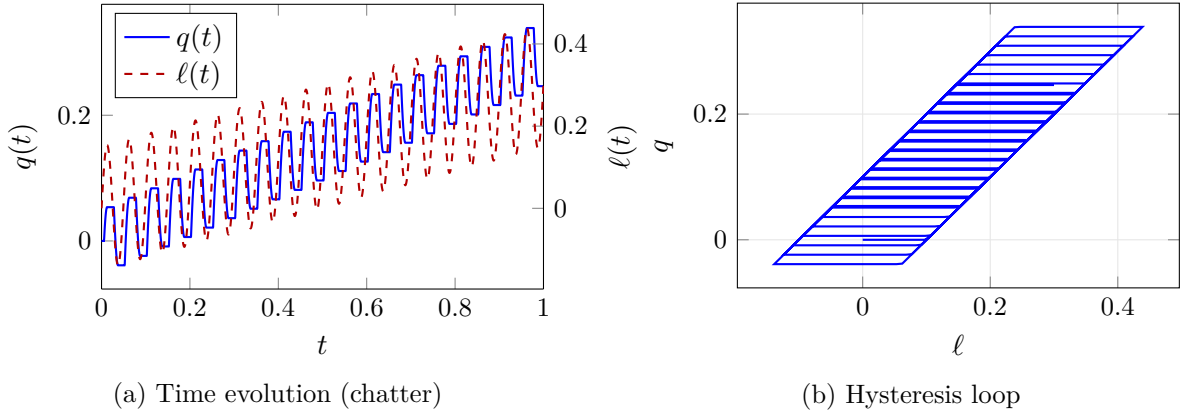
\begin{figure}[t]
  \centering
  \begin{subfigure}[c]{0.47\textwidth}
    \centering
    \begin{tikzpicture}
      \begin{axis}[width=\textwidth,height=0.72\textwidth,
        xmin=0,xmax=1,axis y line*=left,xlabel={$t$},ylabel={$q(t)$},
        tick label style={font=\small}]
        \addplot[thick,blue] table[x=t,y=q]{figures/fig_time_osc.dat};
      \end{axis}
      \begin{axis}[width=\textwidth,height=0.72\textwidth,
        xmin=0,xmax=1,axis y line*=right,axis x line=none,ylabel={$\ell(t)$},
        tick label style={font=\small},legend pos=north west,
        legend cell align=left]
        \addlegendimage{thick,blue}\addlegendentry{$q(t)$}
        \addplot[thick,red!70!black,dashed] table[x=t,y=ell]{figures/fig_time_osc.dat};
        \addlegendentry{$\ell(t)$}
      \end{axis}
    \end{tikzpicture}
    \caption{Time evolution (chatter)}
  \end{subfigure}
  \hfill
  \begin{subfigure}[c]{0.47\textwidth}
    \centering
    \begin{tikzpicture}
      \begin{axis}[width=\textwidth,height=0.72\textwidth,
        xlabel={$\ell$},ylabel={$q$},grid=both,
        grid style={gray!20},tick label style={font=\small}]
        \addplot[thick,blue] table[x=ell,y=q]{figures/fig_loop_osc.dat};
      \end{axis}
    \end{tikzpicture}
    \caption{Hysteresis loop}
  \end{subfigure}
  \caption{Sub-step loading~\eqref{eq:osc_loading} (well-resolved
  numerical solution, $h/\tau=0.005$): the state chatters as the
  oscillation crosses the yield surface, producing a serrated loop.}
  \label{fig:time_osc}
\end{figure}

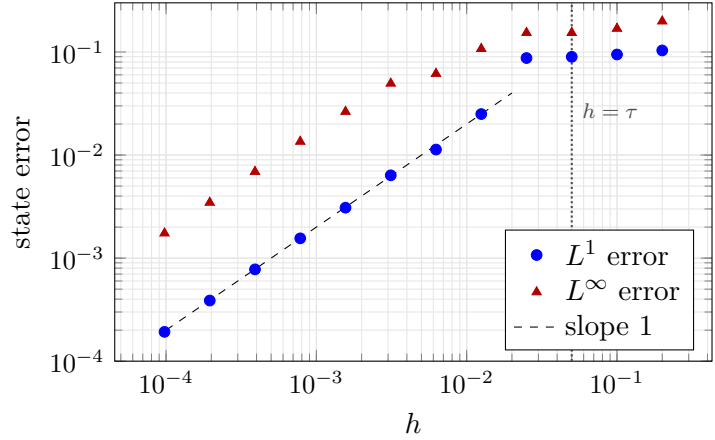
\begin{figure}[t]
  \centering
  \begin{tikzpicture}
    \begin{loglogaxis}[width=0.6\textwidth,height=0.4\textwidth,
      xlabel={$h$},ylabel={state error},ymin=1e-4,ymax=3e-1,
      legend pos=south east,legend cell align=left,
      grid=both,grid style={gray!20},tick label style={font=\small}]
      \addplot[only marks,mark=*,blue] table[x=h,y=err_L1]{figures/fig_convergence_osc.dat};
      \addplot[only marks,mark=triangle*,red!70!black] table[x=h,y=err_Linf]{figures/fig_convergence_osc.dat};
      \addplot[dashed,black,domain=1e-4:2e-2]{2.0*x};
      \draw[densely dotted,gray!60!black,thick]
        (axis cs:0.05,1e-4) -- (axis cs:0.05,3e-1)
        node[pos=0.7,anchor=west,font=\scriptsize] {$h=\tau$};
      \legend{$L^1$ error,$L^\infty$ error,slope $1$}
    \end{loglogaxis}
  \end{tikzpicture}
  \caption{Convergence of the discrete state in norm under the sub-step
  loading~\eqref{eq:osc_loading} (log--log): a plateau at $O(1)$ persists
  for $h\gtrsim\tau$, collapsing onto the first-order interpolation regime
  of \cref{prop:interp} only once $h$ resolves $\tau$.}
  \label{fig:conv_osc_state}
\end{figure}

\begin{figure}[t]
  \centering
  \begin{tikzpicture}
    \begin{loglogaxis}[width=0.6\textwidth,height=0.4\textwidth,
      xlabel={$h$},ylabel={$\sum_k e_k$ (energy defect)},ymin=1e-4,ymax=3e-1,
      legend pos=south west,legend cell align=left,
      grid=both,grid style={gray!20},tick label style={font=\small}]
      \addplot[only marks,mark=square*,teal!70!black] table[x=h,y=energy_defect]{figures/fig_convergence_osc.dat};
      \addplot[dashed,black,domain=1e-4:2e-2]{26.0*x};
      \draw[densely dotted,gray!60!black,thick]
        (axis cs:0.05,1e-4) -- (axis cs:0.05,3e-1)
        node[pos=0.92,anchor=west,font=\scriptsize] {$h=\tau$};
      \legend{energy defect,slope $1$}
    \end{loglogaxis}
  \end{tikzpicture}
  \caption{First-order energy consistency under the sub-step
  loading~\eqref{eq:osc_loading} (log--log): the global energy-balance
  defect stays below the proven bound $C\,h$ of
  \cref{thm:global_consistency} at every $h$---smallest exactly where the
  state error (\cref{fig:conv_osc_state}) is largest---and settles onto
  its first-order asymptote once $h$ resolves $\tau$.}
  \label{fig:conv_osc_energy}
\end{figure}
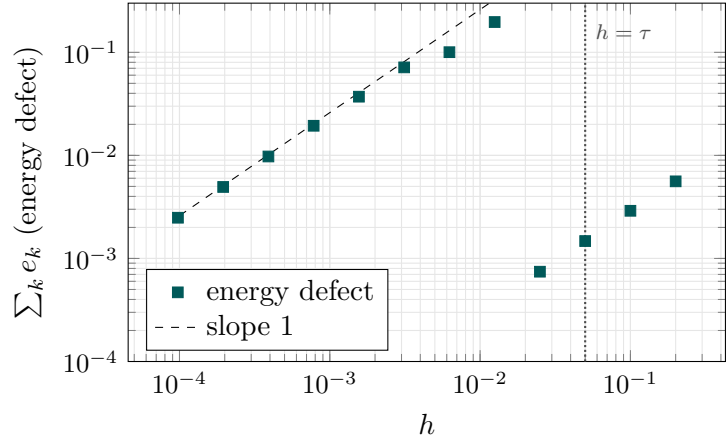

Because $\Q\ni q\mapsto E(q,S)$ is convex on this benchmark, the
incremental minimisation~\eqref{eq:Ikstar} has a unique critical point
and the distinction between a \emph{global} and a \emph{local} solver
of it is moot: there is nothing here to expose the difference between
energetic and balanced-viscosity solutions. The next benchmark restores
that distinction.

\subsection{A bistable double well}\label{sec:benchmark2}

\subsubsection{Problem set-up} \label{sec:setup}

Let $\Q=\R$ and
\begin{equation}\label{eq:benchmark2}
  F(q)=
  \begin{cases}
    \tfrac12 k(q+a)^2, & q\le0,\\
    \tfrac12 k(q-a)^2, & q>0,
  \end{cases}
  \qquad
  \ell(S)=\ell_\infty\bigl(1-e^{-\lambda S}\bigr),
  \qquad
  \Psi(\dot q)=\rho\,|\dot q|,
\end{equation}
with $k,a,\rho>0$. This is a genuinely \emph{bistable} landscape: two
wells at $q=\mp a$ (where $F=0$), separated by a corner at $q=0$ (where
$F$ is continuous but not differentiable, $F=\tfrac12ka^2$). The
corner is a caricature of a smooth barrier chosen because $g(q)=F'(q)=k(q+a)$ for $q\le0$ and
$k(q-a)$ for $q>0$ is \emph{piecewise linear}, so every yield condition
below inverts elementarily. $F$ is
continuous and coercive of order $p=2$ (\ref{A:W}), and $E(\cdot,S)$
has exactly two local minima for every $S$, so
\cref{thm:BV_existence,thm:BV_uniqueness} apply: a BV solution exists
and, in this scalar case, is unique.

\paragraph{The jump is unavoidable.} Under monotonically increasing
loading, a transition between the two wells of \emph{any} non-convex
bistable energy is necessarily discontinuous: continuously tracking a
local equilibrium would require $g$ to increase throughout, but $g$
is not monotone between the wells (here it jumps down, from $+ka$ to
$-ka$, at the corner). There is no smooth analogue, for this benchmark,
of the reversible regime of \cref{sec:benchmark}.

\paragraph{Energetic solution (closed form).} Because
\eqref{eq:Ikstar} minimises \emph{globally} over $\Q$, it computes the
energetic solution (\cref{thm:convergence}), whose switching instant is
set by a Maxwell-type comparison of $E(q,S)+\Dr(q(t),q)$ across both
wells, not merely by the loss of stability of the current branch. Fixing
the initial datum $q(0)=-a$ and comparing the two wells' own energies
under the loading, one finds (\cref{app:experiments}) a remarkably
simple, closed-form threshold, independent of $k$ and $a$:
\begin{equation}\label{eq:l_jump_energetic}
  \ell_{\mathrm{jump}} = \rho,
  \qquad
  q^-=-a \ \longmapsto\ q^+=+a,
\end{equation}
i.e.\ the state remains locked at the bottom of the initial well until
$\ell(S(t))$ reaches $\rho$, at which instant it jumps directly to the
bottom of the other well. Thereafter, on the new branch,
\begin{equation}\label{eq:play_exact2}
  q(t)=
  \begin{cases}
    -a, & t\le t_{\mathrm{jump}},\\[2pt]
    \bigl(\ell(t)-\rho\bigr)/k+a, & t_{\mathrm{jump}}<t\le T/2,\\[2pt]
    \min\bigl\{q_{\max},\,(\ell(t)+\rho)/k+a\bigr\}, & T/2<t\le T,
  \end{cases}
\end{equation}
with $q_{\max}=(\ell_{\max}-\rho)/k+a$ and $t_{\mathrm{jump}}$ obtained
in closed form by inverting $\ell(S(t))=\rho$ on the ascending branch.
As in \cref{sec:benchmark}, the residual state
$q_{\mathrm{res}}=q(T)=\min\{q_{\max},\,\rho/k+a\}$ exhibits the same
reversible/lock-in dichotomy governed by $\rho$ against $\ell_{\max}$.

\begin{remark}[Contrast with the balanced-viscosity threshold]
\label{rem:energetic_vs_bv}
Energetic solutions are known to be able to switch earlier than is
physically natural, precisely because condition (S) compares the
current state against \emph{every} other state, not merely against
nearby ones; this is exactly the motivation for balanced-viscosity
solutions given in \cref{sec:bv}. The present benchmark makes that
gap concrete and quantitative: the left branch's own local stability,
$|f(q,S)|\le\rho$ evaluated \emph{on that branch alone}, persists until
$\ell = ka+\rho$, strictly later than the energetic
threshold~\eqref{eq:l_jump_energetic} (they coincide only in the
degenerate case $a=0$, no barrier). The incremental
scheme~\eqref{eq:Ikstar}, as formally stated, minimises over all of
$\Q$ and therefore targets the energetic solution; approximating the
physically-selected BV solution instead would require restricting the
minimisation to a neighbourhood of $q_k$ (a local, return-map-style
solve, as is standard in computational plasticity) together with the
vanishing-viscosity jump construction of \cref{def:BV}, which we do
not carry out numerically here.
\end{remark}

\subsubsection{Numerical validation of the jump}\label{sec:regimes2}

We fix $k=1$, $a=0.15$, $\rho=0.10$, and reuse the loading protocol of
\cref{sec:regimes} ($\ell_\infty=0.5$, $\lambda=1$, $S_{\max}=5$,
$T=1$), so that $\ell_{\mathrm{jump}}=0.10<\ell_{\max}\approx0.497$ and
the jump occurs at $t_{\mathrm{jump}}\approx0.0678$.
\Cref{fig:time2} shows the numerical time evolution and the resulting
hysteresis loop: the state is locked at $-a$, jumps discontinuously to
$+a$, continues to flow on the new branch up to the peak load, and
partially recovers on unloading. \Cref{tab:val2} confirms the
closed-form threshold, jump target, peak and residual are all
reproduced by the integrator.

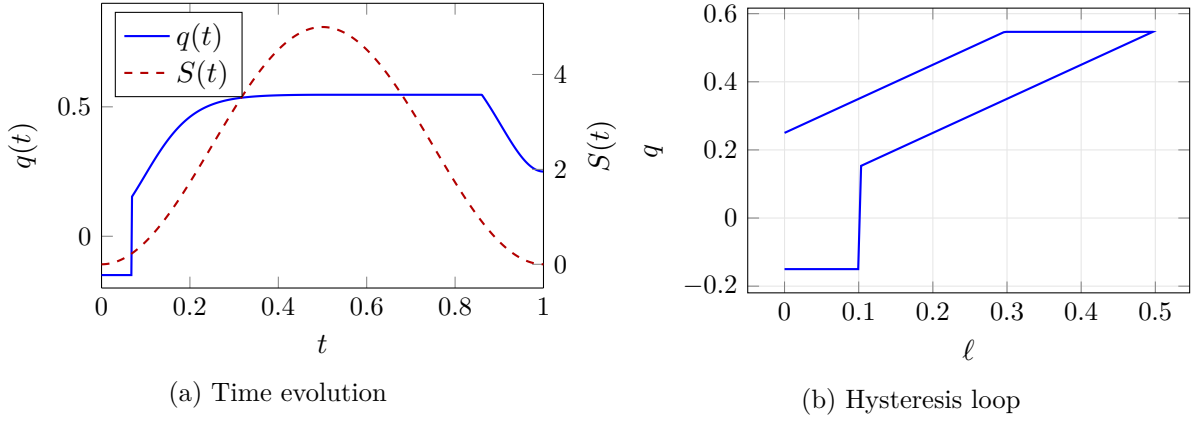
\begin{figure}[t]
  \centering
  \begin{subfigure}[c]{0.47\textwidth}
    \centering
    \begin{tikzpicture}
      \begin{axis}[width=\textwidth,height=0.72\textwidth,
        xmin=0,xmax=1,axis y line*=left,xlabel={$t$},ylabel={$q(t)$},
        ymin=-0.2, ymax=0.9,
        tick label style={font=\small}]
        \addplot[thick,blue] table[x=t,y=q]{figures/fig_time_bistable.dat};
      \end{axis}
      \begin{axis}[width=\textwidth,height=0.72\textwidth,
        xmin=0,xmax=1,axis y line*=right,axis x line=none,ylabel={$S(t)$},
        tick label style={font=\small},legend pos=north west,
        legend cell align=left]
        \addlegendimage{thick,blue}\addlegendentry{$q(t)$}
        \addplot[thick,red!70!black,dashed] table[x=t,y=S]{figures/fig_time_bistable.dat};
        \addlegendentry{$S(t)$}
      \end{axis}
    \end{tikzpicture}
    \caption{Time evolution}
  \end{subfigure}
  \hfill
  \begin{subfigure}[c]{0.47\textwidth}
    \centering
    \begin{tikzpicture}
      \begin{axis}[width=\textwidth,height=0.72\textwidth,
        xlabel={$\ell$},ylabel={$q$},grid=both,
        grid style={gray!20},tick label style={font=\small}]
        \addplot[thick,blue] table[x=ell,y=q]{figures/fig_loop_bistable.dat};
      \end{axis}
    \end{tikzpicture}
    \caption{Hysteresis loop}
  \end{subfigure}
  \caption{Bistable benchmark: (a) the state jumps discontinuously from
  $-a$ to $+a$ once $\ell(S(t))$ reaches $\rho$, then flows smoothly;
  (b) the corresponding loop in the $(\ell,q)$ plane, with a vertical
  segment at the jump.}
  \label{fig:time2}
\end{figure}

\begin{table}[t]
  \centering
  \caption{Closed-form~\eqref{eq:l_jump_energetic}--\eqref{eq:play_exact2}
  vs.\ computed values ($N=4000$).}
  \label{tab:val2}
  \begin{tabular}{lcc}
    \toprule
    quantity & exact & numerical \\
    \midrule
    $\ell_{\mathrm{jump}}$ & $0.10000$ & $0.10063$ $^{\ast}$ \\
    $q^+$ (post-jump)      & $0.15000$ & $0.15064$ $^{\ast}$ \\
    $q_{\max}$ (peak)      & $0.54663$ & $0.54663$ \\
    $q_{\mathrm{res}}$ (residual) & $0.25000$ & $0.25000$ \\
    \bottomrule
  \end{tabular}
  \vspace{2pt}
  \begin{minipage}{0.85\textwidth}
    \small $^\ast$Evaluated at the first grid node past
    $t_{\mathrm{jump}}$; the $O(h)$ discrepancy is precisely the
    subject of \cref{sec:conv2}.
  \end{minipage}
\end{table}

\subsubsection{Convergence and consistency across a jump}\label{sec:conv2}

The bistable benchmark lets us examine two \emph{logically distinct}
questions that the smooth benchmark of \cref{sec:conv_exp} could not
separate: the \emph{convergence} of the discrete state to the exact one,
and the \emph{consistency} of the discrete energy balance. Only the
second is covered by a rate theorem; we treat them in turn.

As in \cref{sec:conv_exp}, at every grid node the return map solves the
same threshold comparison as the exact solution, using the same previous
state $q_k$, so $q^h(t_i)=q(t_i)$ exactly at the nodes; the single jump
does not spoil node-wise exactness. The discretisation error is thus
entirely an \emph{interpolation} effect: the natural piecewise-constant
representation $q^h(t)=q^h(t_i)$ for $t\in[t_i,t_{i+1})$ differs from the
continuously varying exact solution by a zeroth-order-hold error, and
locates the jump only to within one step. \Cref{fig:traj2} makes this
visible: as $h$ is refined the staircase $q^h$ tightens onto the exact
curve and the numerically resolved jump migrates towards
$t_{\mathrm{jump}}\approx0.068$ (at $t=0.10,\,0.08,\,0.07$ for
$h=0.1,\,0.04,\,0.01$).

\begin{figure}[t]
  \centering
  \begin{tikzpicture}
    \begin{axis}[width=0.72\textwidth,height=0.42\textwidth,
      xmin=0,xmax=1,xlabel={$t$},ylabel={$q(t)$},
      legend pos=south east,legend cell align=left,
      grid=both,grid style={gray!20},tick label style={font=\small}]
      \addplot[const plot,blue,thick] table[x=t,y=q]{figures/fig_traj_bistable_c.dat};
      \addplot[const plot,orange!85!black] table[x=t,y=q]{figures/fig_traj_bistable_m.dat};
      \addplot[const plot,teal!70!black,thin] table[x=t,y=q]{figures/fig_traj_bistable_f.dat};
      \addplot[black,very thick] table[x=t,y=q]{figures/fig_traj_bistable_exact.dat};
      \legend{$h=0.1$,$h=0.04$,$h=0.01$,exact}
    \end{axis}
  \end{tikzpicture}
  \caption{Visual convergence of the piecewise-constant discrete
  trajectory $q^h$ (const-plot staircases) to the exact
  solution~\eqref{eq:play_exact2} as the step size $h$ is refined. The
  jump is resolved to within one step and migrates towards the exact
  instant $t_{\mathrm{jump}}$.}
  \label{fig:traj2}
\end{figure}
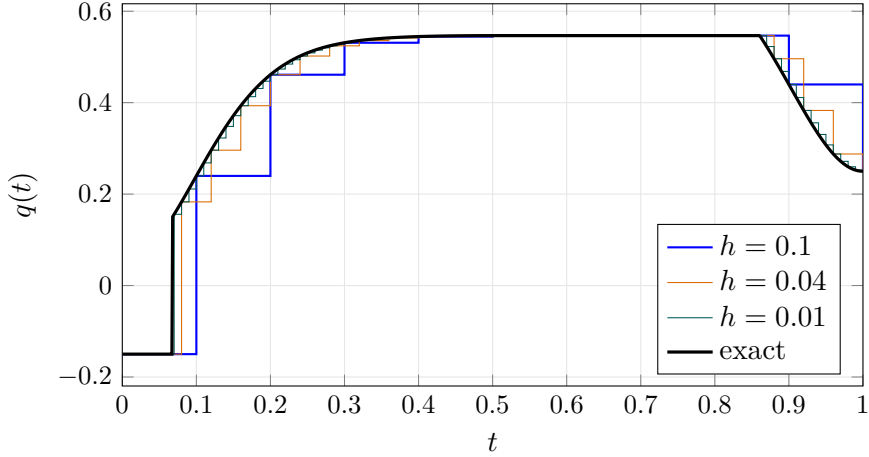

\paragraph{Convergence in norm.} \Cref{thm:convergence} guarantees that
$q^h\to q$, but provides \emph{no general rate}. \Cref{fig:conv2_state}
reports two norms of the state error against a fine, $h$-independent
sampling of the exact solution --- the time-averaged ($L^1$) error and
the sup-norm error restricted to times at least $O(h)$ away from
$t_{\mathrm{jump}}$. Both are first order, tracking the reference
slope-one line, with no qualitative penalty from the jump. As in the
previous example this rests on nodal exactness: the return map
reproduces the exact solution at the nodes even across the jump, so the
norm error is a pure zeroth-order-hold interpolation error, which is
$O(h)$ for a piecewise-Lipschitz solution with an isolated jump by
\cref{prop:interp} (with $J$ the size of the single jump). We stress
again (\cref{rem:nodal}) that this is genuine of the present scalar
setting, not a general a priori rate.

\begin{figure}[t]
  \centering
  \begin{tikzpicture}
    \begin{loglogaxis}[width=0.6\textwidth,height=0.4\textwidth,
      xlabel={$h$},ylabel={state error},
      legend pos=north west,legend cell align=left,
      grid=both,grid style={gray!20},tick label style={font=\small}]
      \addplot[only marks,mark=*,blue] table[x=h,y=err_L1]{figures/fig_convergence_bistable.dat};
      \addplot[only marks,mark=triangle*,red!70!black] table[x=h,y=err_Linf_away]{figures/fig_convergence_bistable.dat};
      \addplot[dashed,black,domain=1.5e-4:2e-2]{0.5*x};
      \legend{$L^1$ error,$L^\infty$ error,slope $1$}
    \end{loglogaxis}
  \end{tikzpicture}
  \caption{Convergence of the discrete state in norm (log--log). Both
  the $L^1$ error and the sup-norm error away from the jump decay
  \emph{empirically} at first order; \cref{thm:convergence} guarantees
  convergence but not this rate.}
  \label{fig:conv2_state}
\end{figure}
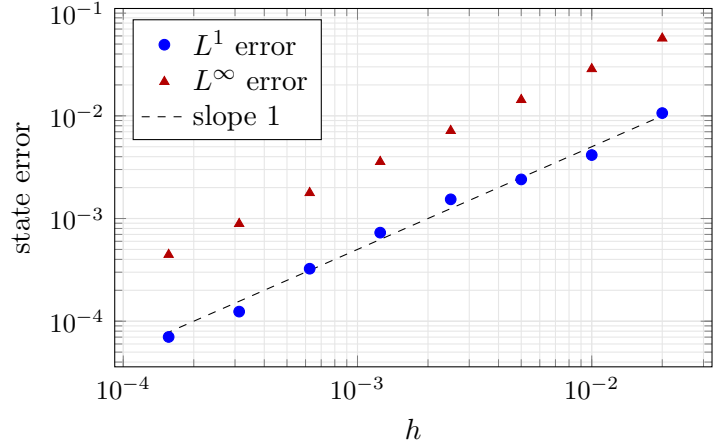

\paragraph{First-order energy consistency.} The energy balance, by
contrast, \emph{does} obey a rate theorem. On the smooth branches
$q$ is Lipschitz, so \cref{thm:global_consistency} applies and bounds
the accumulated truncation error by $O(h)$. The single jump, at which
$q$ is only of bounded variation and the theorem's regularity hypothesis
fails, contributes no additional order: the energetic transition is
\emph{dissipation-neutral} --- at the Maxwell threshold $\ell=\rho$ one
has $I(-a)=I(+a)=a\rho$ (\cref{app:experiments}), equivalently
$E(q^-,S)-E(q^+,S)=\Dr(q^-,q^+)$, so the exact jump releases exactly the
energy it dissipates and carries zero viscous correction. The discrete
over-dissipation at the jump step is therefore controlled by the
$O(h)$ offset between $\ell$ at the jump node and $\rho$. Consequently
the global energy-balance defect $\sum_k e_k=(E_0+\mathcal W_h)-(E_N+\Diss_h)$
decays as $O(h)$ (\cref{fig:conv2_energy}), tracking the slope-one line
predicted by \cref{thm:global_consistency}.

\begin{figure}[t]
  \centering
  \begin{tikzpicture}
    \begin{loglogaxis}[width=0.6\textwidth,height=0.4\textwidth,
      xlabel={$h$},ylabel={$\sum_k e_k$ (energy defect)},
      legend pos=north west,legend cell align=left,
      grid=both,grid style={gray!20},tick label style={font=\small}]
      \addplot[only marks,mark=square*,teal!70!black] table[x=h,y=energy_defect]{figures/fig_convergence_bistable.dat};
      \addplot[dashed,black,domain=1.5e-4:2e-2]{1.0*x};
      \legend{energy defect,slope $1$}
    \end{loglogaxis}
  \end{tikzpicture}
  \caption{First-order energy consistency across the jump (log--log):
  the global energy-balance defect decays as $O(h)$, as guaranteed by
  \cref{thm:global_consistency} on the smooth arcs and preserved across
  the dissipation-neutral energetic jump.}
  \label{fig:conv2_energy}
\end{figure}
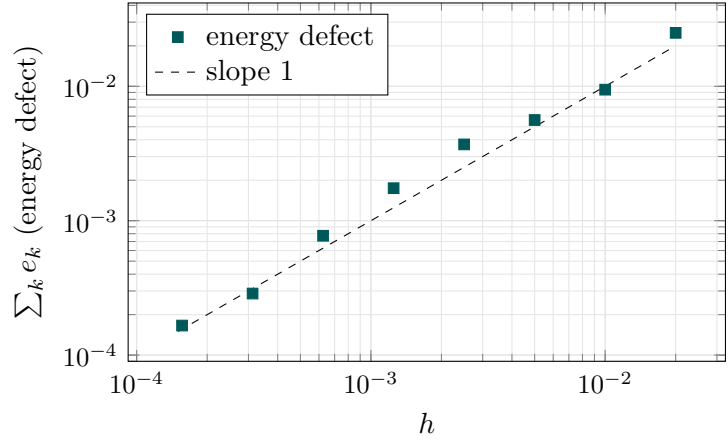

\subsection{A nonlinear convex example}\label{sec:benchmark3}

\subsubsection{Problem set-up}\label{sec:setup3}

The first two benchmarks share a feature that flatters the return map:
the force $g=F'$ is linear (\cref{sec:benchmark}) or piecewise linear
(\cref{sec:benchmark2}), so the yield condition $g(q)=\ell\mp\rho$
inverts in closed form and the corrector is a single explicit formula.
To probe the method in its general form we now take a genuinely
\emph{nonlinear} landscape, whose force has no elementary inverse, so
that the return map must solve the yield condition by an
\emph{iterative} (Newton) scheme at every dissipative step. Let
$\Q=\R$ and
\begin{equation}\label{eq:benchmark3}
  F(q)=\tfrac12 q^2+\cosh q-1,\qquad
  g(q)=F'(q)=q+\sinh q,\qquad
  \Psi(\dot q)=\rho\,|\dot q|,
\end{equation}
with $\rho>0$. The energy is smooth and \emph{uniformly convex},
$F''(q)=1+\cosh q\ge2$, so \ref{A:W}--\ref{A:coerc} hold and, by
\cref{thm:uniqueness}, the energetic solution is unique; the
transcendental force $g(q)=q+\sinh q$ admits no closed-form inverse.

We validate the integrator by the \emph{method of manufactured
solutions}. Rather than prescribe the loading and solve for the state
(which would require $g^{-1}$), we prescribe a smooth, piecewise-monotone
state $q(t)$---a lock--flow--lock--flow--lock cycle rising from $q=0$ to
$q=1$ and back, with $C^1$ (smoothstep) ramps---and take as loading the
$\ell(t)$ that makes it the exact play-operator response, obtained by
\emph{forward evaluation} of $g$ alone:
\begin{equation}\label{eq:manufactured}
  \ell(t)=g(q(t))+\rho \ \ (\text{forward flow}),
  \qquad
  \ell(t)=g(q(t))-\rho \ \ (\text{backward flow}),
\end{equation}
with $\ell$ sweeping inside the elastic band $[g(q)-\rho,g(q)+\rho]$
during lock (the loading is here specified directly in time, equivalently
$S(t)=t$). By construction $q(t)$ is the exact solution for this $\ell$;
the integrator, given only $\ell$ at the nodes, must recover it by solving
$g(q_{k+1})=\ell(t_{k+1})\mp\rho$ with a safeguarded Newton iteration.

\subsubsection{Numerical validation}\label{sec:regimes3}

We fix $\rho=0.3$, $T=1$, with the state locked on $[0,0.15]$, flowing up
on $[0.15,0.45]$, locked at the peak on $[0.45,0.55]$, flowing down on
$[0.55,0.85]$ and locked on $[0.85,1]$. \Cref{fig:time3} shows the
numerical time evolution and the resulting hysteresis loop.
\Cref{tab:val3} reports the validation: the return map reproduces the
manufactured solution at the nodes to machine precision, and does so
through a genuine iterative solve---an average of three Newton iterations
per dissipative step, converging to a residual
$|g(q)-\ell\pm\rho|\sim10^{-14}$---confirming that the method does not
rely on any closed-form invertibility of $g$.

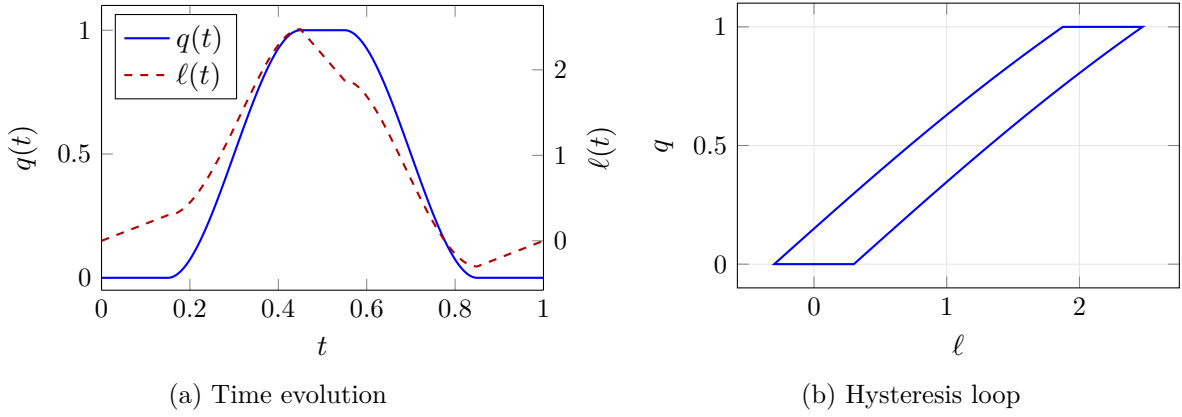
\begin{figure}[t]
  \centering
  \begin{subfigure}[t]{0.47\textwidth}
    \centering
    \begin{tikzpicture}
      \begin{axis}[width=\textwidth,height=0.72\textwidth,
        xmin=0,xmax=1,axis y line*=left,xlabel={$t$},ylabel={$q(t)$},
        ymin=-0.05,ymax=1.1,tick label style={font=\small}]
        \addplot[thick,blue] table[x=t,y=q]{figures/fig_time_nonlinear.dat};
      \end{axis}
      \begin{axis}[width=\textwidth,height=0.72\textwidth,
        xmin=0,xmax=1,axis y line*=right,axis x line=none,ylabel={$\ell(t)$},
        tick label style={font=\small},legend pos=north west,
        legend cell align=left]
        \addlegendimage{thick,blue}\addlegendentry{$q(t)$}
        \addplot[thick,red!70!black,dashed] table[x=t,y=ell]{figures/fig_time_nonlinear.dat};
        \addlegendentry{$\ell(t)$}
      \end{axis}
    \end{tikzpicture}
    \caption{Time evolution}
  \end{subfigure}
  \hfill
  \begin{subfigure}[t]{0.47\textwidth}
    \centering
    \begin{tikzpicture}
      \begin{axis}[width=\textwidth,height=0.72\textwidth,
        xlabel={$\ell$},ylabel={$q$},grid=both,
        grid style={gray!20},tick label style={font=\small}]
        \addplot[thick,blue] table[x=ell,y=q]{figures/fig_loop_nonlinear.dat};
      \end{axis}
    \end{tikzpicture}
    \caption{Hysteresis loop}
  \end{subfigure}
  \caption{Nonlinear convex benchmark ($g=q+\sinh q$): (a) manufactured
  time evolution of the state $q(t)$ and the loading $\ell(t)$ that
  induces it; (b) the corresponding hysteresis loop in the $(\ell,q)$
  plane. All curves are the numerical solution.}
  \label{fig:time3}
\end{figure}

\begin{table}[t]
  \centering
  \caption{Validation of the nonlinear benchmark ($N=4000$): node-wise
  accuracy of the manufactured solution and behaviour of the Newton
  solver in the return map.}
  \label{tab:val3}
  \begin{tabular}{lc}
    \toprule
    quantity & value \\
    \midrule
    nodal state error $\max_i|q^h(t_i)-q(t_i)|$ & $4.8\times10^{-14}$ \\
    average Newton iterations / dissipative step & $3.0$ \\
    maximum Newton residual $|g(q)-\ell\pm\rho|$ & $9.7\times10^{-14}$ \\
    \bottomrule
  \end{tabular}
\end{table}

\subsubsection{Convergence and consistency}\label{sec:conv3}

As in \cref{sec:conv2}, the return map is exact at the nodes---the yield
condition $g(q_{k+1})=\ell(t_{k+1})\mp\rho$ is solved (now
\emph{iteratively}) at each node independently of $h$---so the
discretisation error is again an interpolation effect and, on this smooth
benchmark, entirely of zeroth-order-hold type. \Cref{fig:traj3} shows the
piecewise-constant staircases tightening onto the exact curve as $h$ is
refined. We separate, as before, the two logically distinct questions.

\Cref{fig:conv3_state} reports the convergence of the state in norm:
against a fine, $h$-independent sampling of the exact solution, both the
time-averaged ($L^1$) error and the sup-norm error decay cleanly at first
order. As in the previous examples this is the interpolation error under
nodal exactness (\cref{rem:nodal}, \cref{prop:interp})---here with no
jump, so $J=0$ and the sup-norm error is $O(h)$ throughout---and not a
general a priori rate. \Cref{fig:conv3_energy} reports the
first-order energy consistency: the global energy-balance defect decays
as $O(h)$, tracking the slope-one line predicted by
\cref{thm:global_consistency}, which here applies throughout since $q$ is
Lipschitz and there is no jump. Crucially, the Newton iteration count
stays bounded---about three per step---under refinement, so the nonlinear
solve adds only a constant factor to the cost and does not degrade the
convergence order.

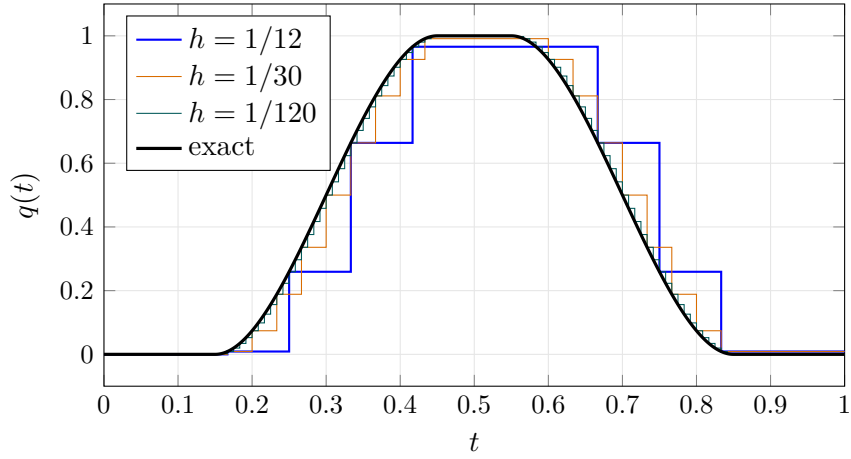
\begin{figure}[t]
  \centering
  \begin{tikzpicture}
    \begin{axis}[width=0.72\textwidth,height=0.42\textwidth,
      xmin=0,xmax=1,xlabel={$t$},ylabel={$q(t)$},
      legend pos=north west,legend cell align=left,
      grid=both,grid style={gray!20},tick label style={font=\small}]
      \addplot[const plot,blue,thick] table[x=t,y=q]{figures/fig_traj_nl_c.dat};
      \addplot[const plot,orange!85!black] table[x=t,y=q]{figures/fig_traj_nl_m.dat};
      \addplot[const plot,teal!70!black,thin] table[x=t,y=q]{figures/fig_traj_nl_f.dat};
      \addplot[black,very thick] table[x=t,y=q]{figures/fig_traj_nl_exact.dat};
      \legend{$h=1/12$,$h=1/30$,$h=1/120$,exact}
    \end{axis}
  \end{tikzpicture}
  \caption{Visual convergence of the piecewise-constant discrete
  trajectory $q^h$ (const-plot staircases) to the exact manufactured
  solution as the step size $h$ is refined.}
  \label{fig:traj3}
\end{figure}

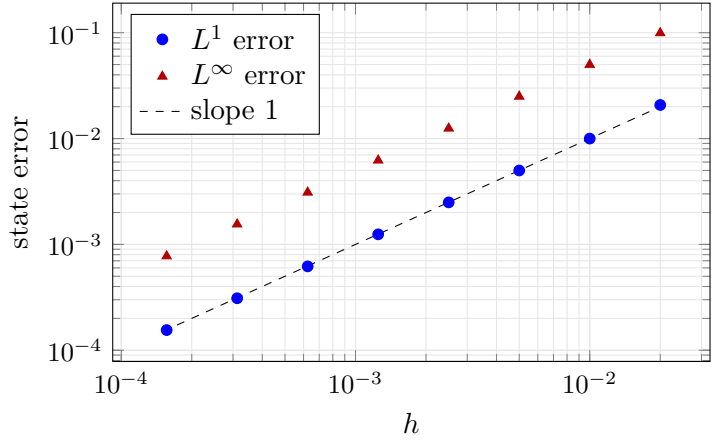
\begin{figure}[t]
  \centering
  \begin{tikzpicture}
    \begin{loglogaxis}[width=0.6\textwidth,height=0.4\textwidth,
      xlabel={$h$},ylabel={state error},
      legend pos=north west,legend cell align=left,
      grid=both,grid style={gray!20},tick label style={font=\small}]
      \addplot[only marks,mark=*,blue] table[x=h,y=err_L1]{figures/fig_convergence_nonlinear.dat};
      \addplot[only marks,mark=triangle*,red!70!black] table[x=h,y=err_Linf_away]{figures/fig_convergence_nonlinear.dat};
      \addplot[dashed,black,domain=1.5e-4:2e-2]{1.0*x};
      \legend{$L^1$ error,$L^\infty$ error,slope $1$}
    \end{loglogaxis}
  \end{tikzpicture}
  \caption{Convergence of the discrete state in norm (log--log) for the
  nonlinear benchmark. Both the $L^1$ error and the sup-norm error decay
  \emph{empirically} at first order; the transcendental force, solved by
  Newton, does not degrade the rate.}
  \label{fig:conv3_state}
\end{figure}

\begin{figure}[t]
  \centering
  \begin{tikzpicture}
    \begin{loglogaxis}[width=0.6\textwidth,height=0.4\textwidth,
      xlabel={$h$},ylabel={$\sum_k e_k$ (energy defect)},
      legend pos=north west,legend cell align=left,
      grid=both,grid style={gray!20},tick label style={font=\small}]
      \addplot[only marks,mark=square*,teal!70!black] table[x=h,y=energy_defect]{figures/fig_convergence_nonlinear.dat};
      \addplot[dashed,black,domain=1.5e-4:2e-2]{8.6*x};
      \legend{energy defect,slope $1$}
    \end{loglogaxis}
  \end{tikzpicture}
  \caption{First-order energy consistency for the nonlinear benchmark
  (log--log): the global energy-balance defect decays as $O(h)$, as
  guaranteed by \cref{thm:global_consistency} (applicable throughout,
  since $q$ is Lipschitz with no jump).}
  \label{fig:conv3_energy}
\end{figure}
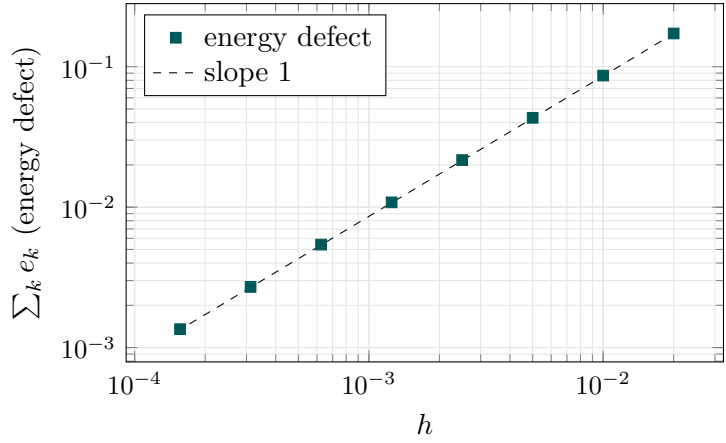

%% file: sections/discussion.tex
\section{Discussion and outlook}\label{sec:discussion}

\subsection{The rate-independent limit and its regime of validity}
\label{sec:ri_validity}

Rate independence is a modelling idealisation, and we adopt it
deliberately rather than as an established fact. It is the
\emph{quasi-static limit} of epigenetic evolution---the regime in which
the chromatin state either equilibrates quickly relative to the
variation of its micro-environment, or remains effectively frozen except
when a threshold is crossed---and it plays here the role that perfect
plasticity or the play operator play in mechanics: questionable in
general, yet capturing the essential phenomenology (threshold
activation, memory, and hysteresis) while yielding a robust,
thermodynamically consistent and predictive skeleton. The price of the
idealisation is explicit. The pure rate-independent model has \emph{no}
intrinsic time scale, so it cannot represent rate-dependent relaxation,
and it predicts the \emph{absence of ratcheting}: repeating an identical
loading cycle leaves the same residual state each time, and the maximal
damage and residual scar depend only on the extreme values of the
loading, not on its rhythm. This is a sharp, falsifiable prediction: an
experimentally observed dependence of the epigenetic scar on the
\emph{frequency} of an intermittent stimulus at fixed amplitude would
signal a genuine viscous (rate-dependent) contribution, quantified by
the parameter~$\varepsilon$ of the viscous extension discussed in
\cref{sec:beyond_ri}.

\subsection{Relation to the quasi-potential landscape paradigm}
\label{sec:quasipotential}

The dominant mathematical realisation of Waddington's epigenetic
landscape~\cite{Waddington1957,Huang2012} constructs it as a
\emph{quasi-potential} from the stationary distribution of a stochastic
gene-regulatory dynamics~\cite{Wang2008,Ferrell2012}. A central lesson
of that programme is that gene-network dynamics is generically a
\emph{non-equilibrium, non-gradient} system: the deterministic drift
does not derive from a single potential, and a non-vanishing rotational
probability flux, sustained by energy input and responsible for
thermodynamic dissipation, is an essential ingredient. This is in
apparent tension with the present framework, which postulates a stored
energy~$E$ and a gradient-type (energetic) evolution. We do not claim
that the full gene-regulatory dynamics is a gradient flow. Rather, RIE
describes the slow, hysteretic \emph{epigenetic-mark} subsystem in the
quasi-static limit, where an effective energetic description is natural
and where the memory of the loading history, not the fast rotational
flux, dominates the response. In this sense RIE is complementary to the
quasi-potential picture: it trades the kinetic and flux content of the
stochastic landscape for an exact thermodynamic bookkeeping (the first
and second laws hold by construction), analytical thresholds, and the
robustness of the variational structure. Clarifying, for a given
biological system, when the quasi-static energetic reduction is
legitimate---and how the effective energy relates to the underlying
quasi-potential---is an important open question.

\subsection{Relation to kinetic bistable-switch models}
\label{sec:kinetic_models}

A large body of work models epigenetic and cell-fate decisions with
kinetic bistable switches: systems of ordinary differential equations
with cooperative feedback whose bifurcation structure produces
bistability and hysteresis~\cite{Huang2007,Ferrell2012}. Coupled reversible and
irreversible switches have been used to explain the
epithelial--mesenchymal transition and its hysteresis~\cite{Tian2013,
CeliaTerrassa2018}. The hysteresis produced by such models is
\emph{kinetic}---the loop depends on the rate at which the control
parameter is swept---whereas the hysteresis of RIE is
\emph{rate-independent}: the loop is a geometric object whose enclosed
area equals the dissipation per cycle, independent of the rhythm of the
stimulus (\cref{sec:hysteresis}). The two viewpoints are complementary.
Kinetic models resolve the molecular circuitry and its time scales at
the cost of many rate constants and of evolution laws that are, in
practice, prescribed by hand; RIE fixes instead the \emph{structure} of
the evolution from three thermodynamic ingredients, obtaining
closed-form thresholds and unconditional consistency, at the cost of
suppressing the kinetics.

\subsection{Beyond rate independence}\label{sec:beyond_ri}

The framework is the rate-independent member of a broader family, and
the route out of it is already present in the paper. Relaxing the
$1$-homogeneity of the dissipation potential---for instance by adding a
quadratic term $\tfrac{\varepsilon}{2}\|\dot{\bs q}\|^2$---turns the
Biot inclusion into a rate-dependent (viscous), doubly-nonlinear
gradient flow. This is exactly the viscous regularisation used to
\emph{construct} the balanced-viscosity solutions of \cref{sec:bv}: the
rate-independent model is the $\varepsilon\to0$ skeleton of a
one-parameter family, and a finite $\varepsilon$ reintroduces relaxation,
frequency dependence and ratcheting~\cite{mielke2012bv,MielkeRoubicek}.
A systematic study of the finite-$\varepsilon$ regime, and the
identification of $\varepsilon$ from frequency-resolved experiments, is a
natural next step. Two further extensions are of direct biological
relevance: \emph{vectorial} state variables ($n\ge2$), needed to
represent several coupled epigenetic marks, for which the transition
path at a switch becomes genuinely path-dependent and the
balanced-viscosity selection is essential; and \emph{non-symmetric}
dissipation potentials, $\Psi(-\dot q)\ne\Psi(\dot q)$, which model the
direction-dependent resistance of methylation versus demethylation and
the associated asymmetry of the forward and reverse thresholds.

\subsection{Interpretation, identifiability and the thermodynamic
statements}\label{sec:interpretation}

Two caveats delimit the reading of the framework. First, the
ingredients $F$, $\bs\ell$ and $\Psi$ are \emph{effective},
phenomenological objects, not directly measured energies; their value is
that they are, in principle, identifiable from macroscopic observables,
since the closed-form thresholds and residual scar
(\cref{sec:experiments}) express them in terms of quantities accessible to a suitably
designed loading--unloading experiment. Second, the first- and
second-law statements of \cref{sec:thermo} are exact \emph{within} the
thermodynamics of internal variables of the model; they should be read
as an effective, structural thermodynamics---with $E$ a Helmholtz-type
effective free energy and $\mathcal D$ an effective entropy
production---rather than as a literal accounting of molecular energetics.
Their role is to guarantee that any RIE model, whatever the choice of
landscape and dissipation, is internally consistent with the balance
principles, which is precisely the robustness that motivates the
framework.

This effective status suggests, finally, a deeper reading. RIE need not
be seen as a mere simplification: the rate-independent, quasi-static
model may be the \emph{leading-order macroscopic reduction} of a
genuinely more complex microscale kinetics, obtained by coarse-graining
or information compression. From this viewpoint the potentials $F$,
$\bs\ell$ and $\Psi$ would be \emph{determined}, rather than posited, as
the effective quantities that survive the reduction---much as, in recent
structure-preserving coarse-graining of nonequilibrium particle systems,
the coarse variables inherit a dissipative, history-dependent effective
dynamics that is thermodynamically consistent by
construction~\cite{Hernandez2026}. In the metriplectic (GENERIC)
description of dissipative systems, RIE occupies the purely dissipative,
rate-independent corner, while the viscous family of
\cref{sec:beyond_ri} interpolates towards the full
reversible--irreversible dynamics. Casting RIE as the first-order block
of such a coarse-grained theory---and learning its potentials from
finer-scale data---is, to us, a particularly promising long-term
direction opened by this work.

%% file: sections/conclusions.tex
\section{Conclusions}\label{sec:conclusions}

We have proposed Rate-Independent Epigenetics (RIE), a framework that
models epigenetic change as rate-independent dissipative evolution
driven by a single triple $(\Q,E,\Psi)$: an energy landscape of
chromatin configurations, a micro-environmental loading, and a
$1$-homogeneous dissipation potential encoding the resistance to
remodelling. Once the triple is fixed and an energetic principle is
postulated, the governing equations---the energetic formulation and the
corresponding Biot inclusion---follow systematically, and their solutions
satisfy, a priori, estimates that coincide with the first and second
laws of thermodynamics. Threshold activation, memory and residual state
upon unloading---the defining phenomenology of epigenetic marks---thus
emerge from a thermodynamically consistent variational structure rather
than from ad-hoc evolution laws. We established the well-posedness of
RIE models (existence, uniqueness under convexity, and existence and
scalar uniqueness of balanced-viscosity solutions) and a convergent
variational integrator that is first-order consistent in the energy
balance and reduces, in the scalar case, to a two-sided return map,
validated against three closed-form benchmarks.

To the best of our knowledge, this is the first transfer of the theory
of \emph{energetic} and \emph{balanced-viscosity} solutions to
epigenetics: the memory, thresholds and hysteresis of epigenetic marks
are recast as the rate-independent evolution of an internal variable,
with thermodynamic consistency built into the variational structure
rather than imposed a posteriori. Beyond this conceptual transfer, the
paper contributes several self-contained results that are, we believe,
of independent interest: an explicit compatibility lemma underpinning
existence (\cref{lem:compatibility}); the definiteness
hypothesis~\ref{A:coerc} required by the total-variation bound of the
balanced-viscosity construction; the node-exactness of the
return-map integrator on rate-independent branches
(\cref{sec:conv_exp}); and, for the piecewise-quadratic bistable well, a
fully closed-form comparison between the \emph{energetic} jump threshold
$\ell=\rho$ and the \emph{local} (balanced-viscosity) threshold
$\ell=ka+\rho$ (\cref{app:experiments}), a transparent instance of the
non-uniqueness that the balanced-viscosity selection resolves.

The scope and limitations of the rate-independent hypothesis, its
relation to the quasi-potential and kinetic paradigms, and the viscous
and vectorial extensions that lie beyond it are discussed in
\cref{sec:discussion}. The present framework is, above all, the
abstract foundation for a quantitative model: it establishes the basis for applying
RIE to a scalar \emph{bistable} (or, in general, \emph{multi-stable}) landscape in
which the healthy and damaged phenotypes are the two wells, catastrophic
phenotype switches of epithelial--mesenchymal type appear as
balanced-viscosity-selected energetic jumps, and the analytical
thresholds and residual ``epigenetic scars'' are calibrated against
experimental data on hypoxia-driven epigenetic remodelling, offering possibilities for designing therapeutic strategies.

\clearpage

%% file: sections/appendix.tex
%
%

\section{Preliminaries and the governing equation}
\label{app:framework}

Throughout this appendix we use, without further comment, two classical
facts. By \emph{Rademacher's theorem} a Lipschitz function
$u:[0,T]\to\R^n$ is differentiable almost everywhere; and, more
generally, an absolutely continuous $u$ is differentiable a.e., its
derivative $\dot u$ is Lebesgue integrable, and the fundamental theorem
of calculus for the Lebesgue integral holds,
$u(t)-u(s)=\int_s^t\dot u(\eta) \, \dd \eta$. This justifies writing $\dot u$ for the
a.e.\ derivative of every absolutely continuous or Lipschitz function
below.

\subsection{Convex-analytic and metric preliminaries}

\begin{lemma}[Saturation identity and scale invariance]\label{lem:saturation}
Let $\Psi\colon\R^n\to[0,+\infty)$ be convex, normalised
($\Psi(\bs 0)=0$) and positively $1$-homogeneous, and recall
$\E:=\partial\Psi(\bs 0)=\{\bs f:\bs f\cdot\bs w\le\Psi(\bs w)\ \forall\,\bs w\in\R^n\}$.
Then:
\begin{enumerate}[label=\textnormal{(\roman*)},leftmargin=2.4em,itemsep=2pt]
  \item\label{sat:i} for every $\bs v\in\R^n$ and every
  $\bs f\in\partial\Psi(\bs v)$,
  \begin{equation}\label{eq:saturation}
    \bs f\cdot\bs v=\Psi(\bs v);
  \end{equation}
  \item\label{sat:ii} $\bs f\in\partial\Psi(\bs v)$ if and only if
  $\bs f\in\E$ \emph{and} $\bs f\cdot\bs v=\Psi(\bs v)$;
  \item\label{sat:iii} $\partial\Psi(\lambda\bs v)=\partial\Psi(\bs v)$
  for every $\lambda>0$.
\end{enumerate}
\end{lemma}

\begin{proof}
\emph{\ref{sat:i}.} By definition of the convex subdifferential,
$\bs f\in\partial\Psi(\bs v)$ means
\begin{equation}\label{eq:subdiff_def}
  \Psi(\bs w)\ge\Psi(\bs v)+\bs f\cdot(\bs w-\bs v)
  \qquad\forall\,\bs w\in\R^n .
\end{equation}
Take $\bs w=\lambda\bs v$ with $\lambda>0$. Positive $1$-homogeneity
gives $\Psi(\lambda\bs v)=\lambda\Psi(\bs v)$, so~\eqref{eq:subdiff_def}
becomes
\[
  \lambda\Psi(\bs v)\ge\Psi(\bs v)+(\lambda-1)\,\bs f\cdot\bs v,
  \qquad\text{i.e.}\qquad
  (\lambda-1)\bigl(\Psi(\bs v)-\bs f\cdot\bs v\bigr)\ge0 .
\]
For $\lambda>1$ the factor $(\lambda-1)$ is positive, whence
$\Psi(\bs v)\ge\bs f\cdot\bs v$; for $\lambda\in(0,1)$ it is negative,
whence $\Psi(\bs v)\le\bs f\cdot\bs v$. Combining the two
gives~\eqref{eq:saturation}.

\emph{\ref{sat:ii}.} ($\Rightarrow$) If $\bs f\in\partial\Psi(\bs v)$,
then~\eqref{eq:saturation} holds, and substituting
$\bs f\cdot\bs v=\Psi(\bs v)$ into the right-hand side
of~\eqref{eq:subdiff_def} yields $\Psi(\bs w)\ge\bs f\cdot\bs w$ for all
$\bs w$, i.e.\ $\bs f\in\E$. ($\Leftarrow$) Conversely, if $\bs f\in\E$
(so $\bs f\cdot\bs w\le\Psi(\bs w)$ for all $\bs w$) and
$\bs f\cdot\bs v=\Psi(\bs v)$, then for every $\bs w$
\[
  \Psi(\bs v)+\bs f\cdot(\bs w-\bs v)
  =\bs f\cdot\bs v+\bs f\cdot(\bs w-\bs v)
  =\bs f\cdot\bs w\le\Psi(\bs w),
\]
which is exactly~\eqref{eq:subdiff_def}; hence
$\bs f\in\partial\Psi(\bs v)$.

\emph{\ref{sat:iii}.} Let $\lambda>0$. By~\ref{sat:ii},
$\bs f\in\partial\Psi(\bs v)$ iff $\bs f\in\E$ and
$\bs f\cdot\bs v=\Psi(\bs v)$. Multiplying the scalar identity by
$\lambda$ and using $1$-homogeneity,
$\bs f\cdot(\lambda\bs v)=\lambda\,\bs f\cdot\bs v=\lambda\Psi(\bs v)=\Psi(\lambda\bs v)$;
since the membership $\bs f\in\E$ does not involve $\bs v$, applying
\ref{sat:ii} again (now at the point $\lambda\bs v$) gives
$\bs f\in\partial\Psi(\lambda\bs v)$. The two conditions are symmetric
in $\bs v$ and $\lambda\bs v$, so the subdifferentials coincide.
\end{proof}

\begin{proposition}[Linearity and continuity of the dissipation
distance]\label{prop:Dr_linear}
If $\Q$ is convex and $\Psi$ is continuous, convex and positively
$1$-homogeneous, then
\begin{equation}\label{eq:Dr_linear}
  \Dr(\bs q_1,\bs q_2)=\Psi(\bs q_2-\bs q_1)\qquad\forall\,\bs q_1,\bs q_2\in\Q,
\end{equation}
and $\Dr$ is continuous on $\Q\times\Q$. In particular
$\Dr(q_1,q_2)=\rho|q_2-q_1|$ when $\Psi(\dot q)=\rho|\dot q|$.
\end{proposition}

\begin{proof}
\emph{Lower bound.} Let $\gamma\in\mathrm{AC}([0,1];\Q)$ join
$\bs q_1$ to $\bs q_2$. Jensen's inequality for the convex function
$\Psi$ applied to the (probability) average of $\dot\gamma$ over
$[0,1]$ gives
\[
  \int_0^1\Psi(\dot\gamma(s))\dd s
  \ge\Psi\!\Bigl(\int_0^1\dot\gamma(s)\dd s\Bigr)
  =\Psi(\gamma(1)-\gamma(0))=\Psi(\bs q_2-\bs q_1).
\]
Taking the infimum over $\gamma$ yields
$\Dr(\bs q_1,\bs q_2)\ge\Psi(\bs q_2-\bs q_1)$.

\emph{Upper bound.} Since $\Q$ is convex, the straight segment
$\bar\gamma(s)=\bs q_1+s(\bs q_2-\bs q_1)$ lies in $\Q$ and is
admissible, with constant velocity $\dot{\bar\gamma}\equiv\bs q_2-\bs q_1$
and cost $\int_0^1\Psi(\bs q_2-\bs q_1)\dd s=\Psi(\bs q_2-\bs q_1)$.
Hence $\Dr(\bs q_1,\bs q_2)\le\Psi(\bs q_2-\bs q_1)$, and~\eqref{eq:Dr_linear}
follows. Continuity of $\Dr$ is then immediate from the continuity of
$\Psi$ and of $(\bs q_1,\bs q_2)\mapsto\bs q_2-\bs q_1$.
\end{proof}

\begin{proposition}[Metric properties of the dissipation
distance]\label{prop:quasi_distance}
Under \ref{A:Psi} and $\Q$ convex, the dissipation distance satisfies,
for all $\bs q_1,\bs q_2,\bs q_3\in\Q$:
\textnormal{(D1)} $\Dr(\bs q_1,\bs q_2)\ge0$ and $\Dr(\bs q,\bs q)=0$;
\textnormal{(D2)} $\Dr(\bs q_1,\bs q_3)\le\Dr(\bs q_1,\bs q_2)+\Dr(\bs q_2,\bs q_3)$.
Since $\Psi$ is symmetric, $\Dr$ is a pseudodistance; if in addition
$\Psi(\bs v)=0\Rightarrow\bs v=\bs 0$ (in particular under
\ref{A:coerc}), $\Dr$ is a distance.
\end{proposition}

\begin{proof}
By \cref{prop:Dr_linear}, $\Dr(\bs q_1,\bs q_2)=\Psi(\bs q_2-\bs q_1)$;
the four properties are then inherited from those of $\Psi$.
\emph{(D1).} $\Dr(\bs q_1,\bs q_2)=\Psi(\bs q_2-\bs q_1)\ge0$ by
positiveness, and $\Dr(\bs q,\bs q)=\Psi(\bs 0)=0$ by normalisation.
\emph{(D2).} A convex, positively $1$-homogeneous $\Psi$ is subadditive:
$\Psi(\bs a+\bs b)=2\,\Psi\bigl(\tfrac12\bs a+\tfrac12\bs b\bigr)
\le2\bigl(\tfrac12\Psi(\bs a)+\tfrac12\Psi(\bs b)\bigr)=\Psi(\bs a)+\Psi(\bs b)$.
Applying this to $\bs a=\bs q_2-\bs q_1$ and $\bs b=\bs q_3-\bs q_2$,
\[
  \Dr(\bs q_1,\bs q_3)=\Psi(\bs q_3-\bs q_1)
  =\Psi\bigl((\bs q_2-\bs q_1)+(\bs q_3-\bs q_2)\bigr)
  \le\Psi(\bs q_2-\bs q_1)+\Psi(\bs q_3-\bs q_2)
  =\Dr(\bs q_1,\bs q_2)+\Dr(\bs q_2,\bs q_3).
\]
\emph{Symmetry and definiteness.} If $\Psi(-\bs v)=\Psi(\bs v)$ then
$\Dr(\bs q_2,\bs q_1)=\Psi(\bs q_1-\bs q_2)=\Psi(\bs q_2-\bs q_1)=\Dr(\bs q_1,\bs q_2)$,
a pseudodistance; and if $\Psi(\bs v)=0\Rightarrow\bs v=\bs 0$ then
$\Dr(\bs q_1,\bs q_2)=0$ forces $\bs q_2=\bs q_1$, so $\Dr$ separates
points.
\end{proof}

\begin{lemma}[Cumulated dissipation of an absolutely continuous
path]\label{lem:diss_ac}
Let $\bs q\in\mathrm{AC}([0,T];\Q)$ with $\Q$ convex. Then
\begin{equation}\label{eq:diss_ac}
  \Diss_\Psi(\bs q;[0,t])=\int_0^t\Psi(\dot{\bs q}(s))\dd s
  \qquad\text{for every }t\in[0,T],
\end{equation}
so $t\mapsto\Diss_\Psi(\bs q;[0,t])$ is absolutely continuous with
$\tfrac{\dd}{\dd t}\Diss_\Psi(\bs q;[0,t])=\Psi(\dot{\bs q}(t))$ for a.e.\
$t$.
\end{lemma}

\begin{proof}
For a partition $0=t_0<\dots<t_N=t$, \cref{prop:Dr_linear} and the
fundamental theorem of calculus give
\[
  \sum_{i=1}^{N}\Dr(\bs q(t_{i-1}),\bs q(t_i))
  =\sum_{i=1}^{N}\Psi\bigl(\bs q(t_i)-\bs q(t_{i-1})\bigr)
  =\sum_{i=1}^{N}\Psi\!\Bigl(\int_{t_{i-1}}^{t_i}\dot{\bs q}(s)  \dd s \Bigr).
\]
Taking the supremum over all partitions, the left side is
$\Diss_\Psi(\bs q;[0,t])$ by definition, while the supremum of the right
side is $\int_0^t\Psi(\dot{\bs q}(s))\dd s$: the inequality "$\le$" follows from
the subadditivity of $\Psi$ (as in \cref{prop:quasi_distance}) together
with Jensen's inequality on each subinterval, and "$\ge$" from
refining the partition, this being precisely the definition of the
$\Psi$-variation of an absolutely continuous curve. This
proves~\eqref{eq:diss_ac}; absolute continuity of $t\mapsto\int_0^t\Psi(\dot{\bs q}(s))\dd s$
and the stated derivative are then standard, since
$s\mapsto\Psi(\dot{\bs q}(s))\in L^1(0,T)$ (as $\Psi$ is Lipschitz on
bounded sets and $\dot{\bs q}\in L^1$).
\end{proof}

\subsection{Energy estimates}

\begin{proposition}[Coercivity, absolute continuity and power
control]\label{prop:power}
Under \ref{A:W}--\ref{A:S}, the following hold.
\begin{enumerate}[label=\textnormal{(E\arabic*)},leftmargin=2.6em,itemsep=2pt]
  \item\label{E:compact} For every $c\in\R$ and $S\ge0$, the sublevel
  set $\mathcal L_c(S):=\{\bs q\in\Q:E(\bs q,S)\le c\}$ is compact.
  \item\label{E:AC} For each fixed $\bs q\in\Q$, $t\mapsto E(\bs q,S(t))$
  is absolutely continuous on $[0,T]$, with
  $\frac{\dd}{\dd t}E(\bs q,S(t))=-\bs q\cdot\bs\ell'(S(t))\dot S(t)$ for
  a.e.\ $t$.
  \item\label{E:gronwall} There exist constants $c_3,c_4>0$, depending
  only on $c_1,c_2,p,C_\ell,C_{\ell'}$, such that
  $|\partial_S E(\bs q,S)|\le c_3\bigl(E(\bs q,S)+c_4\bigr)$ for all
  $\bs q\in\Q$, $S\ge0$; moreover $E(\bs q,S)+c_4\ge1>0$.
\end{enumerate}
\end{proposition}

\begin{proof}
\emph{\ref{E:compact}.} \textit{Boundedness.} Let
$\bs q\in\mathcal L_c(S)$, i.e.\ $F(\bs q)-\bs q\cdot\bs\ell(S)\le c$.
By \ref{A:W}, \ref{A:ell} and Cauchy--Schwarz,
\begin{equation}\label{eq:coercive_sublevel}
  c_1\|\bs q\|^p-c_2\le F(\bs q)\le c+\bs q\cdot\bs\ell(S)\le c+C_\ell\|\bs q\|,
\end{equation}
so $c_1\|\bs q\|^p-C_\ell\|\bs q\|\le c+c_2$. Since $p>1$, the
left-hand side tends to $+\infty$ as $\|\bs q\|\to\infty$; therefore
$\|\bs q\|$ is bounded uniformly on $\mathcal L_c(S)$. \textit{Closedness.}
If $\bs q_n\in\mathcal L_c(S)$ and $\bs q_n\to\bs q$ in $\Q$ (closed),
continuity of $F$ (\ref{A:W}) and of $\bs\ell$ (\ref{A:ell}, Lipschitz)
give $E(\bs q,S)=\lim_n E(\bs q_n,S)\le c$, so $\bs q\in\mathcal L_c(S)$.
A closed, bounded subset of $\R^n$ is compact by Heine--Borel.

\emph{\ref{E:AC}.} By \ref{A:ell}, $\bs\ell$ is Lipschitz, hence, by
Rademacher's theorem, differentiable a.e.\ with $\|\bs\ell'\|\le C_{\ell'}$;
by \ref{A:S}, $S\in W^{1,1}(0,T)\subset\mathrm{AC}([0,T])$. The
composition $t\mapsto\bs\ell(S(t))$ is therefore absolutely continuous,
with $\frac{\dd}{\dd t}\bs\ell(S(t))=\bs\ell'(S(t))\dot S(t)$ a.e.\ by the
chain rule for Sobolev functions. As $F(\bs q)$ does not depend on $t$,
$t\mapsto E(\bs q,S(t))=F(\bs q)-\bs q\cdot\bs\ell(S(t))$ is absolutely
continuous with the stated derivative; its integrability follows from
$|\bs q\cdot\bs\ell'(S(t))\dot S(t)|\le C_{\ell'}\|\bs q\|\,|\dot S(t)|$
and $\dot S\in L^1$.

\emph{\ref{E:gronwall}.} Since $|\partial_S E(\bs q,S)|=|\bs q\cdot\bs\ell'(S)|\le C_{\ell'}\|\bs q\|$,
it suffices to bound $\|\bs q\|$ by $E(\bs q,S)$. From
\eqref{eq:coercive_sublevel} with $c=E(\bs q,S)$,
\begin{equation}\label{eq:coercive_rearranged}
  c_1\|\bs q\|^p-C_\ell\|\bs q\|\le E(\bs q,S)+c_2 .
\end{equation}
Apply Young's inequality $ab\le a^p/p+b^{p'}/p'$ (with
$\tfrac1p+\tfrac1{p'}=1$) to
$a=(c_1p/2)^{1/p}\|\bs q\|$ and $b=C_\ell/(c_1p/2)^{1/p}$, so that
$ab=C_\ell\|\bs q\|$:
\begin{equation}\label{eq:young}
  C_\ell\|\bs q\|\le\frac{c_1}{2}\|\bs q\|^p+K,
  \qquad K:=\frac{1}{p'}\Bigl(\frac{C_\ell}{(c_1p/2)^{1/p}}\Bigr)^{p'}>0 .
\end{equation}
Substituting~\eqref{eq:young} into~\eqref{eq:coercive_rearranged} and
absorbing,
$\tfrac{c_1}{2}\|\bs q\|^p\le E(\bs q,S)+c_2+K=:E(\bs q,S)+\tilde c_2$;
in particular $E(\bs q,S)+\tilde c_2\ge\tfrac{c_1}{2}\|\bs q\|^p\ge0$,
and
$\|\bs q\|\le(2/c_1)^{1/p}\bigl(E(\bs q,S)+\tilde c_2\bigr)^{1/p}$.
Since $1/p<1$, $r^{1/p}\le r+1$ for all $r\ge0$; applying this to
$r=E(\bs q,S)+\tilde c_2$ gives
\[
  |\partial_S E(\bs q,S)|\le C_{\ell'}\|\bs q\|
  \le C_{\ell'}\Bigl(\tfrac{2}{c_1}\Bigr)^{1/p}\bigl(E(\bs q,S)+\tilde c_2+1\bigr)
  = c_3\bigl(E(\bs q,S)+c_4\bigr),
\]
with $c_3=C_{\ell'}(2/c_1)^{1/p}$ and $c_4=\tilde c_2+1$; and
$E(\bs q,S)+c_4=(E(\bs q,S)+\tilde c_2)+1\ge1$.
\end{proof}

\subsection{The Biot inclusion}

\begin{proof}[Proof of \cref{prop:biot}]
\emph{Statement.} \emph{Let $\bs q\in\mathrm{AC}([0,T];\Q)$ satisfy
\textnormal{(S)} and~\textnormal{(E)} with $E(\cdot,S)$ differentiable.
Then $\bs f(\bs q,S)\in\partial\Psi(\dot{\bs q})$ for a.e.\ $t$
$($equivalently~\eqref{eq:biot}$)$; conversely, if $E(\cdot,S)$ is
convex, \eqref{eq:biot} together with $\bs f(\bs q,S)\in\E$ implies
\textnormal{(S)}--\textnormal{(E)}.}

\smallskip
Let $t$ be a point at which $\dot{\bs q}(t)$ exists (a.e.\ $t$, since
$\bs q\in\mathrm{AC}$) and at which the energy balance~(E) is
differentiable.

\emph{Step 1: local stability.} By the global stability~(S), the state
$\bs q(t)$ minimises $\bs r\mapsto E(\bs r,S(t))+\Dr(\bs q(t),\bs r)$
over $\Q$; in particular it is a local minimiser. Since
$E(\cdot,S(t))$ is differentiable and, by \cref{prop:Dr_linear},
$\Dr(\bs q(t),\bs r)=\Psi(\bs r-\bs q(t))$ has subdifferential
$\partial\Psi(\bs r-\bs q(t))$, which at $\bs r=\bs q(t)$ equals
$\partial\Psi(\bs 0)=\E$, the first-order optimality condition reads
$\bs 0\in\partial_{\bs q}E(\bs q(t),S(t))+\E$, that is,
\begin{equation}\label{eq:local_stab}
  \bs f(\bs q(t),S(t))=-\partial_{\bs q}E(\bs q(t),S(t))\in\E .
\end{equation}

\emph{Step 2: energy identity along the trajectory.} By
\cref{prop:power}\ref{E:AC} the map $t\mapsto E(\bs q(t),S(t))$ is
absolutely continuous, and by \cref{lem:diss_ac} so is
$t\mapsto\Diss_\Psi(\bs q;[0,t])$, with derivative $\Psi(\dot{\bs q})$.
Differentiating the balance~(E),
$E(\bs q(t),S(t))+\Diss_\Psi(\bs q;[0,t])=E(\bs q(0),S(0))+\int_0^t\partial_S E\,\dot S$,
at a.e.\ $t$ gives
\[
  \partial_{\bs q}E(\bs q,S)\cdot\dot{\bs q}+\partial_S E(\bs q,S)\,\dot S
  +\Psi(\dot{\bs q})=\partial_S E(\bs q,S)\,\dot S,
\]
and the terms $\partial_S E\,\dot S$ cancel, leaving
\begin{equation}\label{eq:power_identity}
  \bs f(\bs q,S)\cdot\dot{\bs q}
  =-\partial_{\bs q}E(\bs q,S)\cdot\dot{\bs q}
  =\Psi(\dot{\bs q}).
\end{equation}

\emph{Step 3: conclusion.} By~\eqref{eq:local_stab},
$\bs f\in\E$; by~\eqref{eq:power_identity},
$\bs f\cdot\dot{\bs q}=\Psi(\dot{\bs q})$. \Cref{lem:saturation}\ref{sat:ii}
(with $\bs v=\dot{\bs q}$) then gives
$\bs f\in\partial\Psi(\dot{\bs q})$, i.e.\ the Biot
inclusion~\eqref{eq:biot}.

\emph{Converse.} Assume $E(\cdot,S)$ convex and that both~\eqref{eq:biot}
and $\bs f(\bs q,S)\in\E$ hold a.e. Convexity makes local and global
stability equivalent, so $\bs f\in\E$ is~(S). For~(E): by
\cref{lem:saturation}\ref{sat:i}, \eqref{eq:biot} gives
$\bs f\cdot\dot{\bs q}=\Psi(\dot{\bs q})$ a.e.; integrating and using
\cref{lem:diss_ac} together with \cref{prop:power}\ref{E:AC} reverses
Step~2 and recovers the balance~(E).
\end{proof}

\section{Well-posedness}
\label{app:math}

\subsection{Compatibility}

\begin{lemma}[Compatibility conditions]\label{lem:compatibility}
Under \ref{A:W}--\ref{A:Psi}, let $(t_m,\bs q_m)_{m\in\N}$ be a
\emph{stable sequence}, i.e.\ $\bs q_m\in\mathcal S(t_m)$ and
$\sup_m E(\bs q_m,S(t_m))<\infty$, with $t_m\to t$ in $[0,T]$ and
$\bs q_m\to\bs q$ in $\Q$. Then:
\begin{enumerate}[label=\textnormal{(C\arabic*)},leftmargin=2.6em,itemsep=2pt]
  \item\label{C1} for a.e.\ $t$, $\partial_S E(\bs q_m,S(t_m))\to\partial_S E(\bs q,S(t))$;
  \item\label{C2} $\bs q\in\mathcal S(t)$.
\end{enumerate}
\end{lemma}

\begin{proof}
\emph{\ref{C1}.} At any $t\notin\mathcal N$ (the null set outside which
$\dot S$ exists), $\partial_S E(\bs q,S(t))=-\bs q\cdot\bs\ell'(S(t))\dot S(t)$
is linear and continuous in $\bs q$. Since $S\in W^{1,1}(0,T)\subset C([0,T])$
and $\bs\ell'$ is bounded and (a.e.) continuous, $t_m\to t$ gives
$\bs\ell'(S(t_m))\dot S(t_m)\to\bs\ell'(S(t))\dot S(t)$; combined with
$\bs q_m\to\bs q$,
\[
  \partial_S E(\bs q_m,S(t_m))=-\bs q_m\cdot\bs\ell'(S(t_m))\dot S(t_m)
  \longrightarrow-\bs q\cdot\bs\ell'(S(t))\dot S(t)=\partial_S E(\bs q,S(t)).
\]
\emph{\ref{C2}.} By $\bs q_m\in\mathcal S(t_m)$, for every fixed
$\tilde{\bs q}\in\Q$,
\begin{equation}\label{eq:stability_m}
  E(\bs q_m,S(t_m))\le E(\tilde{\bs q},S(t_m))+\Dr(\bs q_m,\tilde{\bs q}).
\end{equation}
As $m\to\infty$: the left side tends to $E(\bs q,S(t))$ by continuity of
$F$ and $\bs\ell$ and $S(t_m)\to S(t)$; the first right-hand term tends
to $E(\tilde{\bs q},S(t))$ for the same reason; the second tends to
$\Dr(\bs q,\tilde{\bs q})$ by continuity of $\Dr$
(\cref{prop:Dr_linear}). Passing to the limit in~\eqref{eq:stability_m},
$E(\bs q,S(t))\le E(\tilde{\bs q},S(t))+\Dr(\bs q,\tilde{\bs q})$ for
every $\tilde{\bs q}$, i.e.\ $\bs q\in\mathcal S(t)$.
\end{proof}

\subsection{Existence and uniqueness of energetic solutions}

\begin{proof}[Proof of \cref{thm:existence}]
\emph{Statement.} \emph{Under \ref{A:W}--\ref{A:Psi}, for every
$\bs q_0\in\mathcal S(0)$ there exists at least one energetic solution
$\bs q:[0,T]\to\Q$ of bounded variation.}

\smallskip
We verify the hypotheses of the abstract existence theorem
\cite[Thm.~2.1.6]{MielkeRoubicek} for the rate-independent system
$(\Q,E,\Dr)$.
\begin{itemize}[leftmargin=1.4em,itemsep=2pt]
  \item \emph{Dissipation distance.} By \cref{prop:quasi_distance},
  $\Dr$ satisfies (D1)--(D2); since $\Psi$ is symmetric (\ref{A:Psi}),
  $\Dr$ is a genuine (extended) quasidistance, and under \ref{A:coerc}
  a distance --- in all cases enough for the cited theorem.
  \item \emph{Compactness of sublevels.} \ref{A:W} and \ref{A:ell} are
  the hypotheses of \cref{prop:power}\ref{E:compact}, which provides
  the required compactness of $\mathcal L_c(S)$ for all $c,S$.
  \item \emph{Control of the power.} \cref{prop:power}\ref{E:AC}
  gives the absolute continuity of $t\mapsto E(\bs q,S(t))$, and
  \ref{E:gronwall} the bound $|\partial_S E(\bs q,S)|\le c_3(E(\bs q,S)+c_4)$.
  Together, $\bigl|\tfrac{\dd}{\dd t}E(\bs q,S(t))\bigr|\le c_3\bigl(E(\bs q,S(t))+c_4\bigr)|\dot S(t)|$,
  i.e.\ the abstract power-control condition holds with
  $\alpha_E(t)=c_3|\dot S(t)|\in L^1(0,T)$ by \ref{A:S}.
  \item \emph{Compatibility.} \cref{lem:compatibility} supplies (C1)
  (continuity of the reduced power along stable sequences) and (C2)
  (closedness of the stability set), the two compatibility conditions
  required.
  \item \emph{Topology.} $\Q\subset\R^n$ with the Euclidean topology is
  separable and metrizable.
\end{itemize}
All hypotheses being met and $\bs q_0\in\mathcal S(0)$ by assumption,
\cite[Thm.~2.1.6]{MielkeRoubicek} yields an energetic solution
$\bs q:[0,T]\to\Q$ of bounded variation in the sense of
\cref{def:energetic}.
\end{proof}

\begin{proof}[Proof of \cref{thm:uniqueness}]
\emph{Statement.} \emph{Under \ref{A:W}--\ref{A:Psi}, if $F\in C^3(\Q)$
is uniformly convex $(D^2F\succeq\alpha\Id,\ \alpha>0)$ and
$\bs\ell\in C^3$, $S\in C^3([0,T])$, then for every
$\bs q_0\in\mathcal S(0)$ the energetic solution is unique.}

\smallskip
Under uniform convexity of $F$ and the linearity of
$\bs q\mapsto-\bs q\cdot\bs\ell(S)$, the reduced energy $E(\cdot,S)$ is
strictly (indeed uniformly) convex, with
$D^2_{\bs q}E(\bs q,S)=D^2F(\bs q)\succeq\alpha\Id$ for every $S\ge0$.
Strict convexity makes the global stability $\bs q_0\in\mathcal S(0)$
equivalent to the pointwise first-order inclusion
$-\partial_{\bs q}E(\bs q_0,S(0))\in\partial\Psi(\bs 0)$, i.e.\ the
initial state is a stationary point of the incremental potential. By the
strengthened smoothness $\bs\ell\in C^3$, $S\in C^3([0,T])$, the map $E$
is jointly $C^3$ on $[0,T]\times\Q$. These are exactly the hypotheses of
\cite[Thm.~3.4.7]{MielkeRoubicek}, whose Section~3.4.4 establishes that
the energetic solution emanating from $\bs q_0$ is unique.
\end{proof}

\subsection{Balanced-viscosity solutions}

\begin{proof}[Proof of \cref{thm:BV_existence}]
\emph{Statement.} \emph{Under \ref{A:W}--\ref{A:coerc}, for every
$\bs q_0\in\mathcal S(0)$ there exists at least one BV solution in the
sense of \cref{def:BV}.}

\smallskip
\emph{Step 1: the regularised problem.} For $\varepsilon>0$ write
$\Phi_\varepsilon(\bs v):=\Psi(\bs v)+\tfrac{\varepsilon}{2}\|\bs v\|^2$;
then \eqref{eq:regularised} is $\bs 0\in\partial\Phi_\varepsilon(\dot{\bs q}_\varepsilon)+\partial_{\bs q}E(\bs q_\varepsilon,S)$.
As $\Phi_\varepsilon$ is strictly convex, coercive and continuous, and
$\partial_{\bs q}E(\cdot,S)$ is a $C^1$ perturbation, the theory of
doubly-nonlinear evolution inclusions governed by a maximal monotone
dissipation \cite{Brezis1973} yields, for each $\varepsilon>0$, a unique
solution $\bs q_\varepsilon\in W^{1,2}(0,T;\Q)$ with
$\bs q_\varepsilon(0)=\bs q_0$, satisfying the (strict) energy identity
\begin{equation}\label{eq:reg_energy}
  E(\bs q_\varepsilon(t),S(t))+\int_0^t\Bigl[\Psi(\dot{\bs q}_\varepsilon(s))+\varepsilon\|\dot{\bs q}_\varepsilon(s)\|^2\Bigr]\dd s
  =E(\bs q_0,S(0))+\int_0^t\partial_S E(\bs q_\varepsilon,S(s))\dot S(s)\dd s .
\end{equation}

\emph{Step 2: uniform energy and total-variation bounds.} Dropping the
non-negative viscous term and the non-negative dissipation
in~\eqref{eq:reg_energy}, and abbreviating
$\phi_\varepsilon(t):=E(\bs q_\varepsilon(t),S(t))+c_4$ (with $c_4$ from
\cref{prop:power}\ref{E:gronwall}, so $\phi_\varepsilon\ge1$),
\[
  \phi_\varepsilon(t)\le\phi_\varepsilon(0)+\int_0^t|\partial_S E(\bs q_\varepsilon,S(s))|\,|\dot S(s)|\dd s
  \le\phi_\varepsilon(0)+\int_0^t c_3\,\phi_\varepsilon(s)\,|\dot S(s)|\dd s,
\]
using \cref{prop:power}\ref{E:gronwall}. Since $|\dot S|\in L^1$,
Gronwall's inequality gives
$\phi_\varepsilon(t)\le\phi_\varepsilon(0)\exp\!\bigl(c_3\|\dot S\|_{L^1}\bigr)$,
whence $\sup_t E(\bs q_\varepsilon(t),S(t))\le M$ with $M$ independent of
$\varepsilon$; by coercivity \ref{A:W}, $\sup_t\|\bs q_\varepsilon(t)\|\le M'$,
also independent of $\varepsilon$. Now, using \ref{A:coerc}
($\Psi(\bs v)\ge c_\Psi\|\bs v\|$), \cref{lem:diss_ac} and the energy
identity,
\begin{align}
  c_\Psi\,\Var(\bs q_\varepsilon;[0,T])
  &=c_\Psi\int_0^T\|\dot{\bs q}_\varepsilon(s)\|\dd s
  \le\int_0^T\Psi(\dot{\bs q}_\varepsilon(s))\dd s
  =\Diss_\Psi(\bs q_\varepsilon;[0,T]) \notag\\
  &\le E(\bs q_0,S(0))-E(\bs q_\varepsilon(T),S(T))
  +\int_0^T|\partial_S E(\bs q_\varepsilon,S(s))|\,|\dot S(s)|\dd s \notag\\
  &\le E(\bs q_0,S(0))+C_{\min}+C_{\ell'}M'\|\dot S\|_{L^1}
  =:c_\Psi\,C, \label{eq:var_bound_bv}
\end{align}
where $E(\bs q_\varepsilon(T),S(T))\ge-C_{\min}$ by coercivity and
$|\partial_S E|\le C_{\ell'}\|\bs q_\varepsilon\|\le C_{\ell'}M'$. Thus
$\Var(\bs q_\varepsilon;[0,T])\le C$, uniformly in $\varepsilon$.

\emph{Step 3: vanishing-viscosity limit.} By Helly's selection theorem
\cite[Thm.~3.3]{AmbrosioFuscoPallara}, the uniform bounds
$\|\bs q_\varepsilon\|_{L^\infty}\le M'$ and $\Var(\bs q_\varepsilon)\le C$
give a subsequence $\varepsilon_j\to0^+$ and a limit
$\bs q\in\mathrm{BV}([0,T];\Q)$ with $\bs q_{\varepsilon_j}(t)\to\bs q(t)$
for every $t\in[0,T]$. Local stability (BV-S) passes to the limit by
lower semicontinuity of $E(\cdot,S)$ (continuity of $F$ and $\bs\ell$)
and continuity of $\Dr$ (\cref{prop:Dr_linear}). For the work term in
(BV-E), the integrands
$\partial_S E(\bs q_{\varepsilon_j}(s),S(s))\dot S(s)=-\bs q_{\varepsilon_j}(s)\cdot\bs\ell'(S(s))\dot S(s)$
converge pointwise a.e.\ to $\partial_S E(\bs q(s),S(s))\dot S(s)$ and
are dominated, uniformly in $j$, by the integrable function
$s\mapsto C_{\ell'}M'|\dot S(s)|\in L^1(0,T)$; Lebesgue's dominated
convergence theorem gives the convergence of the work integral. The
identification of the viscous correction $\mathcal V$ as the limit of
$\int\varepsilon\|\dot{\bs q}_\varepsilon(s)\|^2 \dd s$ concentrated at the jumps
of $\bs q$ follows from the reparametrised (arc-length) passage to the
limit of \cite[Thm.~3.3.2]{MielkeRoubicek}. Hence $\bs q$ is a BV
solution in the sense of \cref{def:BV}.
\end{proof}

\begin{proof}[Proof of \cref{thm:BV_uniqueness}]
\emph{Statement.} \emph{Under \ref{A:W}--\ref{A:coerc}, if $n=1$ and
$q\mapsto E(q,S)$ has finitely many local minima for each $S$, the BV
solution is unique.}

\smallskip
\emph{Uniqueness of the post-jump state.} At a jump time $t^\ast$, the
pre-jump state $q^-(t^\ast)$ is determined by left-continuity of the
preceding arc. In dimension one, any monotone path from $q^-$ to a
candidate $q^+$ realises the dissipation distance $\Dr(q^-,q^+)=\rho|q^+-q^-|$
(\cref{prop:Dr_linear}), and the balanced-viscosity jump condition
extracted from (BV-E) --- namely that the transition follows the
vanishing-viscosity heteroclinic, i.e.\ the path of steepest descent of
$q\mapsto E(q,S(t^\ast))+\rho|q-q^-|$ --- selects $q^+$ as the first
local minimiser of $E(\cdot,S(t^\ast))$ encountered in the direction of
motion beyond the point where $q^-$ loses local stability. Because
$E(\cdot,S(t^\ast))$ has finitely many local minima, this ``first''
minimiser is unambiguous; hence $q^+(t^\ast)$ is unique.

\emph{Uniqueness of the trajectory.} On any interval free of jumps, the
solution is locally stable and satisfies the flow rule
$f\in\partial\Psi(\dot q)$, i.e.\ the scalar Kuhn--Tucker system
$|f|\le\rho$ with $\dot q\ne0$ only when $|f|=\rho$; this determines
$q(t)$ uniquely as the monotone solution of the algebraic yield
equation $g(q)=\ell(S)\mp\rho$ on the current branch, with the sign
fixed by the direction of loading. Since both the smooth arcs and the
jump targets are uniquely determined and the jump times are the
(isolated) instants at which local stability is lost, the whole
trajectory is unique.
\end{proof}

\section{Variational numerical method}
\label{app:numerical}

\begin{proposition}[Return map]\label{prop:return_map}
Under \ref{A:Psi}, the incremental problem~\eqref{eq:Ikstar} is
equivalent to the inclusion~\eqref{eq:discrete_inclusion}, which
decomposes into two mutually exclusive cases:
\textnormal{(i)} \emph{elastic lock}, $\bs q_{k+1}=\bs q_k$, iff
$-\partial_{\bs q}E(\bs q_k,S_{k+1})\in\E$; and
\textnormal{(ii)} \emph{dissipative flow}, $\bs q_{k+1}\ne\bs q_k$ with
$-\partial_{\bs q}E(\bs q_{k+1},S_{k+1})\in\partial\Psi(\bs q_{k+1}-\bs q_k)\setminus\E$.
In the scalar case $\Psi(\dot q)=\rho|\dot q|$ this is precisely the
return map~\eqref{eq:scalar_return_map}.
\end{proposition}

\begin{proof}
The functional $\bs q\mapsto E(\bs q,S_{k+1})+\Dr(\bs q_k,\bs q)$ is the
sum of the differentiable $E(\cdot,S_{k+1})$ and the convex
$\Dr(\bs q_k,\cdot)$ (\cref{prop:quasi_distance}); hence $\bs q_{k+1}$
minimises it iff
\begin{equation}\label{eq:opt_incremental}
  \bs 0\in\partial_{\bs q}E(\bs q_{k+1},S_{k+1})+\partial_{\bs q}\Dr(\bs q_k,\bs q_{k+1}).
\end{equation}
By \cref{prop:Dr_linear}, $\Dr(\bs q_k,\bs q)=\Psi(\bs q-\bs q_k)$, so
$\partial_{\bs q}\Dr(\bs q_k,\bs q)=\partial\Psi(\bs q-\bs q_k)$ and
\eqref{eq:opt_incremental} is exactly~\eqref{eq:discrete_inclusion}. If
$\bs q_{k+1}=\bs q_k$, then $\bs q_{k+1}-\bs q_k=\bs 0$ and the inclusion
reads $-\partial_{\bs q}E(\bs q_k,S_{k+1})\in\partial\Psi(\bs 0)=\E$,
giving case~(i). If $\bs q_{k+1}\ne\bs q_k$, then
$-\partial_{\bs q}E(\bs q_{k+1},S_{k+1})\in\partial\Psi(\bs q_{k+1}-\bs q_k)$,
which by \cref{lem:saturation}\ref{sat:ii} lies in $\E$ only on its
boundary; case~(ii) records that the force is on $\partial\E$ and drives
a non-zero increment. In the scalar case, $\partial\Psi(v)=\{\rho\}$ for
$v>0$, $[-\rho,\rho]$ for $v=0$, $\{-\rho\}$ for $v<0$; writing the
trial force $f^{\mathrm{tr}}=\ell(S_{k+1})-g(q_k)$, case~(i) is
$|f^{\mathrm{tr}}|\le\rho$, while case~(ii) forces
$\ell(S_{k+1})-g(q_{k+1})=\pm\rho$ with the sign of $q_{k+1}-q_k$, i.e.\
$g(q_{k+1})=\ell(S_{k+1})-\sgn(f^{\mathrm{tr}})\rho$, which
is~\eqref{eq:scalar_return_map}.
\end{proof}

\begin{proposition}[Discrete stability and energy inequality]
\label{prop:discrete_energy}
Let $\{\bs q_k\}_{k=0}^N$ be generated by~\eqref{eq:Ikstar}. Then, for
each $k$,
\begin{align}
  &E(\bs q_{k+1},S_{k+1})\le E(\tilde{\bs q},S_{k+1})+\Dr(\bs q_k,\tilde{\bs q})
  \qquad\forall\,\tilde{\bs q}\in\Q, \label{eq:disc_stab}\\
  &E(\bs q_{k+1},S_{k+1})+\Dr(\bs q_k,\bs q_{k+1})\le E(\bs q_k,S_{k+1}). \label{eq:disc_energy}
\end{align}
\end{proposition}

\begin{proof}
By minimality of $\bs q_{k+1}$ in~\eqref{eq:Ikstar}, for every
$\tilde{\bs q}\in\Q$,
\begin{equation}\label{eq:min_incremental}
  E(\bs q_{k+1},S_{k+1})+\Dr(\bs q_k,\bs q_{k+1})
  \le E(\tilde{\bs q},S_{k+1})+\Dr(\bs q_k,\tilde{\bs q}).
\end{equation}
Since $\Dr(\bs q_k,\bs q_{k+1})\ge0$ (D1), dropping it from the
left-hand side of~\eqref{eq:min_incremental} gives~\eqref{eq:disc_stab}.
Taking $\tilde{\bs q}=\bs q_k$ in~\eqref{eq:min_incremental} and using
$\Dr(\bs q_k,\bs q_k)=0$ (D1) gives~\eqref{eq:disc_energy}.
\end{proof}

\begin{proof}[Proof of \cref{thm:convergence}]
\emph{Statement.} \emph{Under \ref{A:W}--\ref{A:Psi}, with
$\bs q_0\in\mathcal S(0)$ and $\bs q^h$ the piecewise-constant
interpolant of the incremental sequence on a partition of maximal step
$h$, a subsequence $\bs q^{h_j}(t)\to\bs q(t)$ for every $t$, with
$\bs q$ an energetic solution; under uniqueness the whole family
converges.}

\smallskip
Write $\alpha_k:=\int_{t_k}^{t_{k+1}}|\dot S(s)|\dd s$, so that
$\sum_k\alpha_k=\|\dot S\|_{L^1(0,T)}$.

\emph{Step 1: uniform energy and variation bounds.} Summing the
discrete energy inequality~\eqref{eq:disc_energy} over $k=0,\dots,K-1$
telescopes the stored-energy terms to
\begin{equation}\label{eq:summed_energy}
  E(\bs q_K,S_K)+\sum_{k=0}^{K-1}\Dr(\bs q_k,\bs q_{k+1})
  \le E(\bs q_0,S_0)+\sum_{k=0}^{K-1}\bigl[E(\bs q_k,S_{k+1})-E(\bs q_k,S_k)\bigr].
\end{equation}
By \cref{prop:power}\ref{E:AC} and Cauchy--Schwarz, each work increment
is controlled:
\begin{equation}\label{eq:work_incr}
  E(\bs q_k,S_{k+1})-E(\bs q_k,S_k)
  =\int_{t_k}^{t_{k+1}}\!(-\bs q_k\cdot\bs\ell'(S(s)))\dot S(s)\dd s
  \le C_{\ell'}\|\bs q_k\|\,\alpha_k .
\end{equation}
The coercivity \ref{A:W} together with \eqref{eq:young} (as in
\cref{prop:power}\ref{E:gronwall}) gives, with $C_p:=(2/c_1)^{1/p}$ and
$\tilde c_2:=c_2+K$,
\begin{equation}\label{eq:qk_linear}
  \|\bs q_k\|\le C_p\bigl(E(\bs q_k,S_k)+\tilde c_2+1\bigr).
\end{equation}
Substituting~\eqref{eq:work_incr} and~\eqref{eq:qk_linear} into
\eqref{eq:summed_energy} and dropping the non-negative dissipation sum,
and setting $e_k:=E(\bs q_k,S_k)+\tilde c_2+1$,
\begin{equation}\label{eq:disc_gronwall}
  e_K\le e_0+C_{\ell'}C_p\sum_{k=0}^{K-1}e_k\,\alpha_k .
\end{equation}
The discrete Gronwall lemma\footnote{If $e_K\le e_0+C\sum_{k<K}e_k\alpha_k$
with $e_k,\alpha_k\ge0$, then $e_K\le e_0\exp\!\bigl(C\sum_{k<K}\alpha_k\bigr)$;
proved by induction, comparing with the solution of the corresponding
equality.}
applied to~\eqref{eq:disc_gronwall} yields
$e_K\le e_0\exp\!\bigl(C_{\ell'}C_p\|\dot S\|_{L^1}\bigr)=:\widetilde M$,
so $\sup_{k,h}E(\bs q_k,S_k)\le\widetilde M-\tilde c_2-1=:M$, a bound
independent of $h$. Coercivity then gives
$c_1\|\bs q_k\|^p-C_\ell\|\bs q_k\|\le M+c_2$, whence
$\sup_{k,h}\|\bs q_k\|\le M'$ independent of $h$. Finally,
from~\eqref{eq:summed_energy} with $K=N$, using
$E(\bs q_N,S_N)\ge-c_2-C_\ell M'=:-C_{\min}$ and \eqref{eq:work_incr},
\begin{equation}\label{eq:var_bound}
  \Var(\bs q^h;[0,T])=\sum_{k=0}^{N-1}\Dr(\bs q_k,\bs q_{k+1})
  \le E(\bs q_0,S_0)+C_{\min}+C_{\ell'}M'\|\dot S\|_{L^1}=:C_{\mathrm{var}},
\end{equation}
with $C_{\mathrm{var}}$ independent of $h$.

\emph{Step 2: compactness.} By~\eqref{eq:var_bound} and
$\|\bs q^h\|_{L^\infty}\le M'$, Helly's selection theorem
\cite[Thm.~3.3]{AmbrosioFuscoPallara} yields a subsequence $h_j\to0$ and
$\bs q\in\mathrm{BV}([0,T];\Q)$ with $\bs q^{h_j}(t)\to\bs q(t)$ for
every $t\in[0,T]$; $\bs q(t)\in\Q$ since $\Q$ is closed.

\emph{Step 3: verification of \textnormal{(S)}.} Fix $t$ and
$\tilde{\bs q}\in\Q$. Let $k_j$ be the index with $t\in[t_{k_j},t_{k_j+1})$;
then $t_{k_j+1}\to t$ and, by continuity of $S$ ($S\in W^{1,1}\subset C^0$),
$S_{k_j+1}\to S(t)$. The discrete stability~\eqref{eq:disc_stab} at step
$k_j$ reads
$E(\bs q^{h_j}(t),S_{k_j+1})\le E(\tilde{\bs q},S_{k_j+1})+\Dr(\bs q^{h_j}(t),\tilde{\bs q})$.
Passing to the limit $j\to\infty$: the left side tends to
$E(\bs q(t),S(t))$ and the first right-hand term to
$E(\tilde{\bs q},S(t))$ by continuity of $E$ and $\bs q^{h_j}(t)\to\bs q(t)$,
$S_{k_j+1}\to S(t)$; the second term tends to $\Dr(\bs q(t),\tilde{\bs q})$
by continuity of $\Dr$ (\cref{prop:Dr_linear}). Hence
$E(\bs q(t),S(t))\le E(\tilde{\bs q},S(t))+\Dr(\bs q(t),\tilde{\bs q})$
for all $\tilde{\bs q}$, i.e.\ (S).

\emph{Step 4: verification of \textnormal{(E)}.} \textit{Upper bound.}
With $K_j$ such that $t_{K_j}\le t<t_{K_j+1}$, summing
\eqref{eq:disc_energy} over $k<K_j$ and writing the work exactly, as
in~\eqref{eq:work_incr} but with equality, gives
\begin{equation}\label{eq:summed_exact}
  E(\bs q_{K_j},S_{K_j})+\Diss_\Psi(\bs q^{h_j};[0,t_{K_j}])
  \le E(\bs q_0,S_0)+\int_0^{t_{K_j}}\!\partial_S E(\bs q^{h_j}(s),S(s))\dot S(s)\dd s,
\end{equation}
where we used $\bs q^{h_j}(s)=\bs q_k$ on $[t_k,t_{k+1})$. We pass to the
limit in each term. On the left, $E(\bs q_{K_j},S_{K_j})\to E(\bs q(t),S(t))$
by continuity, while lower semicontinuity of the $\Psi$-variation under
pointwise convergence gives
$\Diss_\Psi(\bs q;[0,t])\le\liminf_j\Diss_\Psi(\bs q^{h_j};[0,t])$: for
any fixed partition $0=s_0<\dots<s_M=t$,
$\sum_i\Dr(\bs q^{h_j}(s_i),\bs q^{h_j}(s_{i+1}))\to\sum_i\Dr(\bs q(s_i),\bs q(s_{i+1}))$
by continuity of $\Dr$, and each partial sum is bounded by
$\Diss_\Psi(\bs q^{h_j};[0,t])$; taking $\liminf_j$ and then the
supremum over partitions gives the claim. On the right, the integrands
$\partial_S E(\bs q^{h_j}(s),S(s))\dot S(s)=-\bs q^{h_j}(s)\cdot\bs\ell'(S(s))\dot S(s)$
converge pointwise a.e.\ to $\partial_S E(\bs q(s),S(s))\dot S(s)$ (as
$\bs q^{h_j}(s)\to\bs q(s)$ and $\bs\ell',\dot S$ are fixed) and are
dominated, uniformly in $j$, by the integrable function
$s\mapsto C_{\ell'}M'|\dot S(s)|\in L^1(0,T)$ (using
$\|\bs q^{h_j}(s)\|\le M'$ and $\|\bs\ell'\|\le C_{\ell'}$); Lebesgue's
dominated convergence theorem therefore gives
$\int_0^{t_{K_j}}\partial_S E(\bs q^{h_j},S)\dot S\to\int_0^t\partial_S E(\bs q,S)\dot S$.
Therefore
\begin{equation}\label{eq:E_upper}
  E(\bs q(t),S(t))+\Diss_\Psi(\bs q;[0,t])
  \le E(\bs q_0,S(0))+\int_0^t\partial_S E(\bs q,S(s))\dot S(s)\dd s .
\end{equation}
\textit{Lower bound.} Since $\bs q$ satisfies (S) (Step~3), the standard
lower energy estimate for stable BV curves holds: testing (S) along any
partition $0=s_0<\dots<s_M=t$, summing the resulting inequalities
$E(\bs q(s_{i}),S(s_i))\le E(\bs q(s_{i+1}),S(s_i))+\Dr(\bs q(s_{i}),\bs q(s_{i+1}))$
and passing to the finest partition recovers the reverse inequality
of~\eqref{eq:E_upper} (see \cite[Sec.~2.1]{MielkeRoubicek}).
Combining the two inequalities yields the energy balance~(E). Hence
$\bs q$ is an energetic solution. If the solution is unique
(\cref{thm:uniqueness} or the scalar case of \cref{thm:BV_uniqueness}),
every subsequence has the same limit $\bs q$, so the whole family
$\bs q^h$ converges.
\end{proof}

\begin{proposition}[Local consistency]\label{prop:consistency}
Assume \ref{A:ell} and $S\in W^{1,\infty}(0,T)$, and let $\bs q$ be the
exact energetic solution and $\tau_k$ the truncation
error~\eqref{eq:truncation}.
\begin{enumerate}[label=\textnormal{(\roman*)},leftmargin=2.4em,itemsep=2pt]
  \item\label{cons:weak} If $\bs q\in\mathrm{AC}(0,T;\Q)$, then
  \begin{equation}\label{eq:cons_weak}
    |\tau_k|\le C_{\ell'}\|\dot S\|_{L^\infty}\Bigl(\int_{t_k}^{t_{k+1}}\|\dot{\bs q}(\tau)\|\dd\tau\Bigr)h_k=O(h_k).
  \end{equation}
  \item\label{cons:strong} If $\bs q\in W^{1,\infty}(0,T;\Q)$
  \textnormal{(}equivalently, $\bs q$ Lipschitz\textnormal{)}, then
  \begin{equation}\label{eq:cons_strong}
    |\tau_k|\le\tfrac12 C_{\ell'}\|\dot{\bs q}\|_{L^\infty}\|\dot S\|_{L^\infty}\,h_k^2=O(h_k^2).
  \end{equation}
\end{enumerate}
\end{proposition}

\begin{proof}
By the fundamental theorem of calculus (valid since
$S\in W^{1,\infty}$), the frozen-state discrete work is
$E(\bs q(t_k),S_{k+1})-E(\bs q(t_k),S_k)=\int_{t_k}^{t_{k+1}}\partial_S E(\bs q(t_k),S(s))\dot S(s)\dd s$.
Subtracting this from the continuous work in~\eqref{eq:truncation} and
using $\partial_S E(\bs q,S)=-\bs q\cdot\bs\ell'(S)$,
\begin{align}
  \tau_k
  &=\int_{t_k}^{t_{k+1}}\bigl[\partial_S E(\bs q(s),S(s))-\partial_S E(\bs q(t_k),S(s))\bigr]\dot S(s)\dd s \notag\\
  &=\int_{t_k}^{t_{k+1}}\bigl(\bs q(s)-\bs q(t_k)\bigr)\cdot\bigl(-\bs\ell'(S(s))\bigr)\dot S(s)\dd s . \label{eq:tau_repr}
\end{align}
By Cauchy--Schwarz, \ref{A:ell} ($\|\bs\ell'\|\le C_{\ell'}$) and
$|\dot S|\le\|\dot S\|_{L^\infty}$,
\begin{equation}\label{eq:tau_bound_common}
  |\tau_k|\le C_{\ell'}\|\dot S\|_{L^\infty}\int_{t_k}^{t_{k+1}}\|\bs q(s)-\bs q(t_k)\|\dd s .
\end{equation}
\emph{\ref{cons:weak}.} If $\bs q\in\mathrm{AC}$, then for
$s\in[t_k,t_{k+1}]$, $\|\bs q(s)-\bs q(t_k)\|\le\int_{t_k}^{s}\|\dot{\bs q}\|\dd\tau\le\int_{t_k}^{t_{k+1}}\|\dot{\bs q}\|\dd\tau$;
inserting into~\eqref{eq:tau_bound_common} and integrating the constant
bound over $s\in[t_k,t_{k+1}]$ (length $h_k$) gives~\eqref{eq:cons_weak}.
\emph{\ref{cons:strong}.} If $\bs q\in W^{1,\infty}$, i.e.\ Lipschitz,
then $\|\bs q(s)-\bs q(t_k)\|\le\|\dot{\bs q}\|_{L^\infty}(s-t_k)$, so
\eqref{eq:tau_bound_common} gives
$|\tau_k|\le C_{\ell'}\|\dot S\|_{L^\infty}\|\dot{\bs q}\|_{L^\infty}\int_{t_k}^{t_{k+1}}(s-t_k)\dd s
=\tfrac12 C_{\ell'}\|\dot S\|_{L^\infty}\|\dot{\bs q}\|_{L^\infty}h_k^2$,
which is~\eqref{eq:cons_strong}.
\end{proof}

\begin{proof}[Proof of \cref{thm:global_consistency}]
\emph{Statement.} \emph{Under \ref{A:ell}, $S\in W^{1,\infty}$ and
either $\bs q\in\mathrm{AC}$ or $\bs q\in W^{1,\infty}$, the accumulated
truncation error satisfies $\mathcal E_\tau=\sum_k|\tau_k|\le Ch$.}

\smallskip
Sum the local bounds of \cref{prop:consistency} over $k=0,\dots,N-1$ and
bound each step by $h=\max_k h_k$. In the weak case,
by~\eqref{eq:cons_weak},
\[
  \mathcal E_\tau=\sum_{k}|\tau_k|
  \le C_{\ell'}\|\dot S\|_{L^\infty}\sum_k h_k\int_{t_k}^{t_{k+1}}\|\dot{\bs q}\|\dd\tau
  \le C_{\ell'}\|\dot S\|_{L^\infty}\,h\sum_k\int_{t_k}^{t_{k+1}}\|\dot{\bs q}\|\dd\tau
  = C_{\ell'}\|\dot S\|_{L^\infty}\,h\!\int_0^T\!\|\dot{\bs q}\|\dd\tau,
\]
the last equality because the subinterval integrals partition $[0,T]$.
In the strong case, by~\eqref{eq:cons_strong} and $\sum_k h_k^2\le h\sum_k h_k=Th$,
\[
  \mathcal E_\tau\le\tfrac12 C_{\ell'}\|\dot{\bs q}\|_{L^\infty}\|\dot S\|_{L^\infty}\sum_k h_k^2
  \le\tfrac12 C_{\ell'}\|\dot{\bs q}\|_{L^\infty}\|\dot S\|_{L^\infty}\,T\,h .
\]
In both cases $\mathcal E_\tau\le C\,h=O(h)$, with $C$ as stated in
\cref{thm:global_consistency}.
\end{proof}

\begin{proposition}[State error under nodal exactness]\label{prop:interp}
Let $\bs q\colon[0,T]\to\Q$ have finitely many jump points
$0<t^\ast_1<\dots<t^\ast_m<T$, be Lipschitz with constant $L$ on each
open subinterval between consecutive jumps, and set
$\Delta_j:=\|\bs q(t^{\ast+}_j)-\bs q(t^{\ast-}_j)\|$ and
$J:=\sum_{j=1}^m\Delta_j$. On a partition of maximal step $h$, let
$\bs q^h$ be the right-continuous piecewise-constant interpolant,
$\bs q^h(t)=\bs q^h(t_i)$ for $t\in[t_i,t_{i+1})$, and assume the scheme
is \emph{node-exact}: $\bs q^h(t_i)=\bs q(t_i)$ at every node $t_i$. Then
\begin{enumerate}[label=\textnormal{(\roman*)},leftmargin=2.4em,itemsep=2pt]
  \item\label{interp:sup} for every interval $I\subset[0,T]$ at distance
  greater than $h$ from $\{t^\ast_1,\dots,t^\ast_m\}$,
  $\|\bs q^h-\bs q\|_{L^\infty(I)}\le L\,h$;
  \item\label{interp:L1} $\|\bs q^h-\bs q\|_{L^1(0,T;\Q)}\le(\tfrac12 LT+J)\,h+mL\,h^2$.
\end{enumerate}
In particular $\|\bs q^h-\bs q\|_{L^1}=O(h)$, and the sup-norm error away
from the jumps is $O(h)$.
\end{proposition}

\begin{proof}
Fix a step $[t_i,t_{i+1})$, $h_i:=t_{i+1}-t_i\le h$. Node-exactness gives
$\bs q^h\equiv\bs q(t_i)$ on the step, so
$\|\bs q^h(t)-\bs q(t)\|=\|\bs q(t_i)-\bs q(t)\|$ there.

\emph{Jump-free steps.} If the step contains no jump, $\bs q$ is
$L$-Lipschitz on it, so $\|\bs q(t_i)-\bs q(t)\|\le L(t-t_i)\le L h_i\le Lh$.
This proves~\ref{interp:sup}, because an interval at distance $>h$ from
the jumps meets only jump-free steps (each step has length $\le h$).
Moreover $\int_{t_i}^{t_{i+1}}\|\bs q^h-\bs q\|\dd t\le\int_{t_i}^{t_{i+1}}L(t-t_i)\dd t=\tfrac12 L h_i^2$.

\emph{Steps containing a jump $t^\ast_j$.} The total variation of
$\bs q$ from $t_i$ to $t$ is its Lipschitz part plus the single crossed
jump, so $\|\bs q(t_i)-\bs q(t)\|\le L(t-t_i)+\Delta_j\le L h_i+\Delta_j$,
whence $\int_{t_i}^{t_{i+1}}\|\bs q^h-\bs q\|\dd t\le(L h_i+\Delta_j)h_i\le L h_i^2+\Delta_j\,h$.

\emph{Summation.} Adding over all steps, using $\sum_i h_i=T$,
$h_i\le h$, and that at most $m$ steps contain a jump,
\[
  \|\bs q^h-\bs q\|_{L^1}
  \le\tfrac12 L\sum_i h_i^2+\sum_{j=1}^m\bigl(L h_{i(j)}^2+\Delta_j\,h\bigr)
  \le\tfrac12 L T h+mL\,h^2+J\,h,
\]
which is~\ref{interp:L1}; both bounds are $O(h)$.
\end{proof}

\begin{remark}\label{rem:interp_scope}
\Cref{prop:interp} presumes the node-exactness $\bs q^h(t_i)=\bs q(t_i)$,
which is \emph{not} a generic property of the scheme; it holds in the
scalar benchmarks of \cref{sec:experiments} because, on any interval of
monotone loading, the return map solves the same algebraic yield
condition that the exact solution satisfies at the node, and reproduces
the lock values exactly. It fails whenever the loading carries structure
on a time scale comparable to, or finer than, the step $h$---so that the
exact state changes regime more than once within a step---the simplest
instance being a loading extremum strictly interior to a step, and a more
striking one a sub-step oscillation, for which the state error remains
$O(1)$ until $h$ resolves that scale (\cref{sec:conv_osc}). It fails
likewise, in the vectorial case, where the transition path at a jump is
genuinely path-dependent and the discrete and continuous post-jump states
may differ. A general error rate without node-exactness is not addressed
here.
\end{remark}

\section{Details for the bistable benchmark}
\label{app:experiments}

We prove the energetic jump threshold~\eqref{eq:l_jump_energetic} of
\cref{sec:benchmark2} and record the local-stability threshold quoted in
\cref{rem:energetic_vs_bv}. Throughout, $k,a,\rho>0$ and $\ell\ge0$ is
the (increasing) instantaneous load; we use the incremental
objective~\eqref{eq:Ikstar} with previous state $q_k=-a$,
\begin{equation}\label{eq:I_bistable}
  I(q):=E(q,\ell)+\rho\,|q-(-a)|=F(q)-q\ell+\rho\,|q+a|,
\end{equation}
with $F$ the piecewise-quadratic well~\eqref{eq:benchmark2}. Condition
(S) at $q=-a$ holds iff $-a$ is a global minimiser of $I$ over $\R$.

\begin{proposition}[Energetic jump threshold]\label{prop:jump_energetic}
For~\eqref{eq:benchmark2} with initial datum $q(0)=-a$ and increasing
loading, and provided $\rho<4ka$, the state $q=-a$ is globally stable
exactly on $\{\ell:\ell\le\rho\}$ and loses global stability for
$\ell>\rho$, the comparison being realised by the bottom of the other
well, $q=+a$. Equivalently, the energetic solution jumps at
$\ell_{\mathrm{jump}}=\rho$, from $q^-=-a$ to $q^+=+a$.
\end{proposition}

\begin{proof}
We minimise $I$ in~\eqref{eq:I_bistable} separately on the two branches
$q\le0$ and $q>0$, each of which is strictly convex, then compare their
minima; the switch occurs where the two are equal, a \emph{Maxwell}
(equal-area) construction in the sense of phase-transition
theory~\cite{mielke2002}: the globally stable branch is the one of lower
incremental energy, and the transition occurs where the two branches
are energetically balanced.

\emph{Left branch ($q\le0$).} Here $F(q)=\tfrac12k(q+a)^2$ and
$|q+a|=q+a$ for $q \ge -a$ (note that $I(q)$ is strictly decreasing for $q < -a$, so the minimiser in the left branch must satisfy $q \ge -a$), so
$I(q)=\tfrac12k(q+a)^2-q\ell+\rho(q+a)$ with
$I'(q)=k(q+a)-\ell+\rho$. At $q=-a$, $I'(-a)=\rho-\ell$, which is
$\ge0$ precisely when $\ell\le\rho$; since $I$ is convex on this branch,
$q=-a$ is its minimiser over $\{q\le0\}$ for all $\ell\le\rho$, with
minimum value
\begin{equation}\label{eq:I_left}
  I(-a)=F(-a)+a\ell=a\ell
\end{equation}
(the quadratic term vanishes at the well bottom $q=-a$).

\emph{Right branch ($q>0$).} Here $F(q)=\tfrac12k(q-a)^2$ and
$|q+a|=q+a$, so $I(q)=\tfrac12k(q-a)^2-q\ell+\rho(q+a)$ with unconstrained
minimiser $q_R:=a+(\ell-\rho)/k$ (positive once $\ell>\rho-ka$). Using
$q_R-a=(\ell-\rho)/k$,
\[
  F(q_R)=\frac{(\ell-\rho)^2}{2k},\qquad
  E(q_R,\ell)=F(q_R)-q_R\ell
  =\frac{(\ell-\rho)^2}{2k}-a\ell-\frac{\ell(\ell-\rho)}{k},
\]
and, since $q_R+a=2a+(\ell-\rho)/k$,
\begin{equation}\label{eq:I_right}
  I(q_R)=E(q_R,\ell)+\rho(q_R+a)
  =\frac{(\ell-\rho)^2}{2k}-a\ell-\frac{\ell(\ell-\rho)}{k}+2a\rho+\frac{\rho(\ell-\rho)}{k}.
\end{equation}
The three $\tfrac1k$-terms in~\eqref{eq:I_right} combine as
$\tfrac1k\bigl[\tfrac12(\ell-\rho)^2-\ell(\ell-\rho)+\rho(\ell-\rho)\bigr]
=\tfrac1k(\ell-\rho)\bigl[\tfrac12(\ell-\rho)-(\ell-\rho)\bigr]
=-\tfrac{(\ell-\rho)^2}{2k}$, so
\begin{equation}\label{eq:I_right_simplified}
  I(q_R)=-\frac{(\ell-\rho)^2}{2k}-a\ell+2a\rho .
\end{equation}

\emph{Comparison.} Subtracting~\eqref{eq:I_right_simplified}
from~\eqref{eq:I_left},
\begin{equation}\label{eq:maxwell}
  I(-a)-I(q_R)
  =a\ell-\Bigl(-\frac{(\ell-\rho)^2}{2k}-a\ell+2a\rho\Bigr)
  =2a(\ell-\rho)+\frac{(\ell-\rho)^2}{2k}
  =(\ell-\rho)\Bigl[2a+\frac{\ell-\rho}{2k}\Bigr].
\end{equation}
The bracket $2a+(\ell-\rho)/(2k)$ is increasing in $\ell$; over
$\ell\in[0,\rho]$ its minimum is at $\ell=0$, equal to
$2a-\rho/(2k)$, which is positive precisely when $\rho<4ka$. Under this
hypothesis the bracket is positive throughout $[0,\rho]$, so the sign of
\eqref{eq:maxwell} is that of $(\ell-\rho)$: it is $\le0$ for
$\ell\le\rho$ (with equality only at $\ell=\rho$) and $>0$ for $\ell$
just above $\rho$ (the bracket stays positive there). Consequently
$I(-a)\le I(q_R)$ for all $\ell\in[0,\rho]$, and since $q_R$ is $I$'s
minimiser on the right branch while $-a$ is $I$'s minimiser on the left
branch, $q=-a$ is the global minimiser of $I$ over $\R$ exactly for
$\ell\in[0,\rho]$; it ceases to be so immediately beyond $\ell=\rho$. At
$\ell=\rho$ one has $q_R=a$, so the losing comparison is with $q=+a$,
proving the claim.
\end{proof}

\begin{remark}[Local threshold]\label{rem:local_threshold}
The distinction invoked in \cref{rem:energetic_vs_bv} concerns not the
pre-jump state but \emph{how long the left branch remains locally
admissible}. Tracking the left-branch equilibrium
$q_L(\ell)=-a+(\ell-\rho)/k$ (the stationary point of $I$ on
$\{q\le0\}$ once it moves off $-a$), local stability of that branch
persists as long as its associated force stays within the elastic
interval, i.e.\ until $q_L(\ell)$ reaches the corner $q=0$; this happens
at $\ell=ka+\rho$, strictly larger than the energetic threshold
$\ell=\rho$ whenever $a>0$. A locally (balanced-viscosity) selected
solution would therefore remain on the left branch up to
$\ell=ka+\rho$ before jumping. Identifying the corresponding
BV post-jump state requires the vanishing-viscosity construction of
\cref{sec:bv} and is not carried out numerically here; the incremental
scheme~\eqref{eq:Ikstar}, minimising globally, realises instead the
energetic threshold $\ell=\rho$ of \cref{prop:jump_energetic}.
\end{remark}
\clearpage